\documentclass[screen,acmsmall,nonacm]{acmart}

\newif{\ifsubmission}           %
\submissionfalse

\acmJournal{PACMPL}
\acmVolume{1}
\acmNumber{CONF} %
\acmArticle{1}
\acmYear{2018}
\acmMonth{1}
\acmDOI{} %
\startPage{1}
\setcopyright{none}
\usepackage{main}
\usepackage{mathtools}
\usepackage{mathpartir}
\usepackage{tikz}
\usepackage{tikzit}
\usepackage{accents}
\usepackage{catchfilebetweentags}
\usepackage{agda}
\usepackage{newunicodechar}
\usetikzlibrary{bbox}
\DeclarePairedDelimiter{\ceil}{\lceil}{\rceil}

\usepackage[only,llbracket,rrbracket]{stmaryrd}

\usepackage[color=yellow,textwidth=.75in]{todonotes}
\usepackage{caption}
\usepackage{subcaption}
\usepackage[clock]{ifsym}
\usepackage{stackengine}

\newunicodechar{λ}{\ensuremath{\mathnormal\lambda}}
\newunicodechar{∀}{\ensuremath{\mathnormal\forall}}
\newunicodechar{ℕ}{\ensuremath{\mathnormal{\Nat}}}
\newunicodechar{Π}{\ensuremath{\mathnormal{\Pi}}}
\newunicodechar{≡}{\ensuremath{\mathnormal{\equiv}}}
\newunicodechar{≤}{\ensuremath{\mathnormal{\le}}}
\newunicodechar{∎}{\qed}
\newunicodechar{ʳ}{}
\newunicodechar{ϕ}{\ensuremath{\mathnormal{\phi}}}
\newunicodechar{₀}{\ensuremath{\mathnormal{_{0}}}}

\title{A cost-aware logical framework}

\author{Yue Niu}
\orcid{0000-0003-4888-6042}
\email{yuen@andrew.cmu.edu}

\author{Jonathan Sterling}
\orcid{0000-0002-0585-5564}
\email{jmsterli@cs.cmu.edu}

\author{Harrison Grodin}
\orcid{0000-0002-0947-3520}
\email{hgrodin@andrew.cmu.edu}

\author{Robert Harper}
\orcid{0000-0002-9400-2941}
\email{rwh@cs.cmu.edu}

\affiliation{
  \institution{Carnegie Mellon University}
  \streetaddress{5000 Forbes Ave.}
  \city{Pittsburgh}
  \state{PA}
  \postcode{15213}
  \country{USA}
}

\begin{document}

\begin{abstract}
  We present \textbf{calf}, a \textbf{c}ost-\textbf{a}ware \textbf{l}ogical
  \textbf{f}ramework for studying quantitative aspects of functional programs.
  Taking inspiration from recent work that reconstructs traditional aspects of
  programming languages in terms of a modal account of \emph{phase
  distinctions}, we argue that the cost structure of programs motivates a phase
  distinction between \emph{intension} and \emph{extension}.  Armed with this
  technology, we contribute a synthetic account of cost structure as a
  computational effect in which cost-aware programs enjoy an internal
  noninterference property: input/output behavior cannot depend on
  cost.
  As a full-spectrum dependent type theory, \textbf{calf} presents a unified language
  for programming and specification of both cost and behavior that can be
  integrated smoothly with existing mathematical libraries available in type
  theoretic proof assistants.

  We evaluate \textbf{calf} as a general framework for cost analysis by implementing
  two fundamental techniques for algorithm analysis:
  the \emph{method of recurrence relations} and \emph{physicist's method for amortized analysis}.
  We deploy these techniques on a variety of case studies:
  we prove a tight, closed bound for
  Euclid's algorithm, verify the amortized complexity of batched queues,
  and derive tight, closed bounds for the sequential and \emph{parallel}
  complexity of merge sort, all fully mechanized in the Agda proof assistant.
  Lastly we substantiate the soundness of quantitative reasoning in \textbf{calf} by means of a model construction.
\end{abstract}

\maketitle

\section{Introduction}\label{sec:intro}

Resource usage is an important \emph{intensional} property of program behavior.
With a rich enough type system, extensional properties of programs can be
investigated in the same language as the program is written --- an approach to
verification that has seen much application in type theoretic tools such as
Nuprl, Coq, Agda, and
Idris~\citep{constable:1986,coq:reference-manual,norell:2009,brady:2013}.
Intensional properties such as cost are not typically amenable to such an
internal analysis, in essence because one cannot conventionally have a function
$\mathsf{cost} : \mathsf{bool} \to \nat{}$ that computes the cost of its input
(such a ``function'' could not respect $\beta$-equivalence). To address this
problem, one could instrument programs with their cost, but this
instrumentation must not be allowed to interfere with the input/output behavior
of programs.

\paragraph{A logical framework for cost}

We contribute \calf{}, a \textbf{c}ost-\textbf{a}ware \textbf{l}ogical
\textbf{f}ramework for studying quantitative aspects of functional programs,
combining recent work on cost recurrence
extraction~\citep{kavvos-morehouse-licata-danner:2019} and the
call-by-push-value decomposition of effects in dependent type
theory~\citep{pedrot-tabareau:2020} with recent modal account of phase
distinctions and non-interference of \citet{sterling-harper:2021}.  \calf{}
evinces a phase distinction between extensional and intensional aspects of code
(analogous to the static--dynamic phase distinction of ML languages); then the
incurrence of \emph{cost} is treated as a computational effect that has force
only in the intensional fragment, ensuring that the extensional behavior of a
program does not depend on the costs of its arguments. In particular \calf{}
ensures that one cannot write a function whose \emph{extension} depends
on the cost component of its input.

\paragraph{Evaluation and implementation}

We evaluate the efficacy of \calf{} by formulating two widely used algorithm analysis techniques --
the method of recurrence relations and the physicist's method for amortized analysis --
and deploying them on a variety of case studies. We have also developed an implementation of \calf{}
in the Agda proof assistant that may be accessed through the supplementary materials.
The following results highlight the central contributions
of our case studies, all \emph{fully mechanized} in the Agda proof assistant:
\begin{enumerate}
  \item We prove an asymptotically tight and closed upper bound on the number of primitive arithmetic operations used in Euclid's algorithm for gcd.
  \item We present an amortized analysis of the cost of sequences of operations on batched queues.
  \item We prove asymptotically tight and closed upper bounds on both the sequential and \emph{parallel} complexity of
  insertion sort and merge sort under the comparison cost model.
\end{enumerate}
It is worth emphasizing that the presented case studies all require nontrivial mathematical reasoning,
which usually presents a significant hurdle for fledgling implementations of type theories that
do not come equipped with the vast number of the necessary but well-known theorems.
Our implementation of \calf{} alleviates this pain point by allowing one to directly \emph{import} data types from Agda
whenever they are required for an algorithm, a mechanism that we explain in \cref{sec:implementation}.
At a high level, this design evinces an embedding of the Agda universe of ``pure data types'' into
the effectful metalanguage of \calf{} that enables one to take advantage of Agda's well-developed mathematical library.

\paragraph{Notation}
In this paper we display all mechanized theorems as defined in the implementation
using the typewriter font, \eg \Mech{Calf.Types.Bounded.bound/relax}.

\paragraph{Metatheory and soundness}

In order to be used to study the cost of programs it is important that \calf{}
not derive an equivalence between two programs $M,N:\mathsf{bool}$ that take a
different number of steps to compute. We verify by means of a model
construction that \calf{} does \emph{not} identify computations that incur
different numbers of steps, the first step toward a stronger adequacy theorem
that would establish the equivalence of \calf{}-encodings with traditional
operational cost dynamics \`a la \citet{blelloch-greiner:1995}.

\paragraph{Parallel complexity}

\calf{} is compatible with many interpretations of the cost structure of programs,
among which is the \emph{cost graph} that encodes the \emph{work} (sequential cost)
and \emph{span} (parallel cost) of a program.
Thus \calf{} also supports
reasoning about the \emph{parallel} complexity of programs through an equational
presentation of the profiling semantics of Blelloch and Greiner \citep{blelloch-greiner:1995}.
By focusing on the verification of functional programs, we position \calf{} to
take advantage of the elegant theory of language-based parallelism \citep{blelloch-greiner:1996, greiner-blelloch:1999, spoonhower-blelloch-harper-gibbons:2008}
developed over the past three decades without descending into the space of
imperative, concurrent programs in which the analogous notions are much more complex.

\subsection{Synthetic cost analysis via computational effects}

Although many cost verification frameworks work with the deep embedding of an
object language in an ambient type theory, we take a synthetic approach by
defining \calf{}, a full-spectrum dependent type theory in which cost is
implemented as a primitive \emph{effect}.  This view of cost is inspired by
\citet{kavvos-morehouse-licata-danner:2019}, who define and extract recurrence
relations of functional programs representing their (high-order) cost
structure.

At first glance, cost might seem like a uniform concept that can be applied
indiscriminately to any computation. This first-order view quickly falls part
when we consider the costs of functions, which should be functions themselves.
The question then is to introduce cost into the type theory in such a way that
it can flow through the type structure compositionally. The insight of
\citet{kavvos-morehouse-licata-danner:2019} is to consider the
call-by-push-value (CBPV) structure induced by a certain \emph{cost monad}'s
Eilenberg--Moore category~\citep{levy:2004}, leveraging the fine-grained type
structure of CBPV to assign a compositional meaning to cost at higher type.

A cost monad is just the writer monad $\Alert{\costty\times-}$ for a given
monoid $\prn{\costty,0,+}$; in call-by-push-value, we may interpret a
\emph{value type} by a set and a \emph{computation type} by an algebra for
$\Alert{\costty\times-}$. There is a free-forgetful adjunction $\F\dashv \UU$,
in which the right adjoint projects the carrier set of an algebra and the left
adjoint takes a set $A$ to $\costty\times{A}$.
The semantic situation of the cost monad inspires a CBPV language containing a
single computational effect $\cstep{c}{M}$ that incurs a given cost
$\isof{c}{\costty}$ before computing $M$, such that
$\cstep{c}{\cstep{d}{M}}\equiv \cstep{c+d}{M}$. Indeed, \calf{} is a
dependently typed version of this CBPV language, defined in the style of
\citeauthor{pedrot-tabareau:2020}'s \dcbpv{} calculus.
The principal slogan of CBPV is ``a value is, a computation does,'' which continues
to hold in \calf{}: a value \emph{is} with no associated cost, a computation
\emph{does} using some cost.

Strong monads (such as the cost writer monad from which \calf{} is
abstracted) provide a common source of value--computation dichotomies. While
value types are interpreted as ordinary types, computation types are
interpreted as \emph{algebras} for a given monad, \ie objects of the
Eilenberg--Moore category of the monad.  Given a value type $A$, we can form
the computation type $\F{A}$ whose interpretation is the free algebra.  Hence
in \calf{}, $\F{A}$ classifies \emph{free} computations of $A$ where the costs
of a sequence of computations are aggregated using the monoid structure
\costty{}, and a value \isof{a}{A} is injected into $\F{A}$ via \ret{a} as the
computation yielding $a$ using zero cost.

\subsection{A new phase distinction: behavior \emph{vs.} cost}\label{sec:new-phase-distinction}

The original phase distinction between static (compile-time) and dynamic
(run-time) code arose in the study of module
systems~\citep{harper-mitchell-moggi:1990}, where light-weight static
compatibility is used to facilitate the composition of modules.
The idea was to disallow type-level dependence on dynamic parts of a module.
Recall that a signature of a module consists of declarations of kinds of
constructors (static entities) and types of expressions (dynamic entities), and
a module itself consists of constructors and expressions.
In the case of ML modules the phase distinction associates to every module
functor a function between their static parts (kinds); in this sense, the
static part of a module is entirely independent of the dynamic parts of the
modules it is linked with.

In our setting a different but entirely analogous phase distinction emerges
between extension/behavior and intension/cost.  Every type $A$ in \calf{} can
be thought of as having two parts: an intensional part $\Cl{A}$
characterizing its cost and an extensional part $\Op{A}$ characterizing its
extensional behavior.
We say that a type is (extensional, intensional) if it is isomorphic to its
(extensional, intensional) part. The phase distinction ensures that the
extensional part of a program is independent with the intensional parts of its
arguments. Put another way, the phase distinction of behavior and cost
constitutes a \emph{noninterference} property of intension and extension:

\begin{quote}
  \emph{\textbf{Noninterference.}}
  Any function $\Cl{A}\to\Op{B}$ from an intensional type to an extensional type is
  internally equal to a constant function.
\end{quote}

\subsection{The language of phase distinctions}\label{sec:lang-of-phase-distinctions}

In \calf{} the phase distinction between extension and intension is achieved by
adding a new abstract proposition $\Alert{\ExtOpn}$ called the ``extensional phase''.
Whenever an assumption of type $\ExtOpn$ is present in the context, the cost
structure of programs is rendered trivial; one can think of the fragment of
\calf{} where $\ExtOpn$ is always in the context as a version of ordinary
dependent type theory in which cost is not tracked.  Therefore, the extensional
part of a type $A$ can be recovered as the function space
$\Alert{\Op{A}\coloneqq\prn{\ExtOpn\to{}A}}$. This extensional modality can be used to
state equations between programs that have different costs but identical
input-output behaviors; for instance, we can prove
$\Op\prn{\textit{insertionSort} = \textit{mergeSort}}$,
even thought these algorithms have different costs
under the comparison cost model for sorting. Indeed, the
soundness of \calf{} implies that this equation does not hold outside of $\Op$.

\paragraph{Cost structure as proof-relevance}

As we have pointed out, it makes little sense to think of cost as a property of
an ordinary program, because two such programs may be equal and yet ``have''
different costs. On the other hand we may view cost as a \emph{structure}
(proof-relevant property) over a program, and the projection of ordinary
programs from cost-instrumented programs is implemented in our setting by the
unit of the extensional modality $A\to\Op{A}$. The perspective of cost as
structure is an instance of a more general phenomenon pervading present-day
work in type theory: notions that are ill-posed as
properties of equivalence classes of typed terms can be recovered more
objectively as structures defined over equivalence classes of typed terms, as
in the work of \citet{altenkirch-kaposi:2016:nbe,coquand:2019,sterling-harper:2021,sterling-angiuli:2021}.

\subsection{Quantitative reasoning in \calf{}}\label{sec:quant-reasoning}

The fundamental advantage of \calf{} is that it provides a purely
\emph{equational} approach to quantitative reasoning: a useful \emph{bound} can
be placed on the number of steps engendered in a computation by equating it to
another computation in which the quantity can be observed directly.  For
example, consider a computation $\isof{e}{\F{A}}$; if we can prove that $e =
\mstep{c}{\ret{a}}$ for some value $\isof{a}{A}$, then we are justified to say
that $e$ has cost $c$. This \emph{cost refinement} is captured by the following \calf{} type:
\[
  \mathsf{hasCost}(A,e,c) \coloneqq \Sigma \isof{a}{A}.\, e =_{\F{A}} \mstep{c}{\ret{a}}
\]
In \cref{sec:quantitative-refinements}
we consider more sophisticated refinements that express \emph{cost bounds} rather than precise costs of computations.

There are two things to note in this definition.  First of all, we can see that
cost refinements are not primitive in \calf{}; rather \calf{} is a logical
framework for \emph{defining} quantitative properties such as
$\mathsf{hasCost}$ and then \emph{proving} refinement lemmas about those
properties.
Secondly, our formulation of cost bounds is only meaningful insofar as stepping
is nondegenerate, \ie $\nvdash \mstep{c}{\ret{a}} = \ret{a}$ for any value $a$
and nontrivial cost $c$.
In fact this nondegeneracy property constitutes one of the soundness criteria for
quantitative reasoning in \calf{}, which we prove in \cref{sec:metatheory}.

Under this regime, one proves more refinements as the need arises in a
verification problem or when new forms of computations are introduced.  In
\cref{sec:quantitative-refinements} we present syntax-directed quantitative refinement lemmas
that decompose the bounds on the cost of a computation into bounds on the costs
of its constituent subcomputations.

\subsection{Compositional cost analysis}

As a type theory, \calf{} naturally supports a compositional style of verification.
When localized to quantitative properties of programs, \calf{} evinces the notion
of a \emph{cost signature} \citep{ab-algorithms}, the cost-aware counterpart to the functional
specification of a data structure. In \calf{}, we may specify the quantitative properties of a
data structure by using \emph{cost-aware} dependent functions $\qpi{(\isof{a}{A})}{B}{c}$,
an application of the $\mathsf{hasCost}$ refinement from \cref{sec:quant-reasoning}:
\begin{align*}
  \qpi{(\isof{a}{A})}{B}{c} &= \Sigma \isof{f}{\Pi(A,B)}.\, \Pi \isof{a}{A}.\, \mathsf{hasCost}(B(a), f(a), c(a))
\end{align*}
Thus an element of \qpi{(\isof{a}{A})}{B}{c} is a function $f$ along with a proof that it
satisfies the cost specification $c$ on all instances.

To see this connective in action, consider clients Alice and Bob who both require a data structure
to manipulate graphs.
Alice may request a structure satisfying the left
signature in \cref{fig:cost-sigs}, indicating that they would like edge insertion and membership to both be
logarithmic in the number of vertices $\mathsf{n}$.
On the other hand, Bob's algorithm needs constant time edge membership, but is not so sensitive to changes to the graph.
This requirement is captured by the right signature in \cref{fig:cost-sigs}.

Fortunately, both programmers can be supplied with suitable implementations: edge sets for Alice
and adjacency matrices for Bob. Although somewhat artificial,
this example shows that \calf{} is able to formalize the notion of a cost signature as used by
\citet{ab-algorithms}, paving the way to verified, cost-aware development of large-scale programs.

\begin{figure}
  \begin{minipage}{0.45\textwidth}
    \begin{align*}
      \textit{sig}&\\
      \quad \mathsf{G} &: \mathsf{Type}\\
      \quad \mathsf{n} &: G \to \Nat\\
      \quad \mathsf{insertEdge} &: \mathsf{edge} \to (\qpi{G}{G}{\log\circ\, \mathsf{n}})\\
      \quad \mathsf{isEdge} &: \mathsf{edge} \to (\qpi{G}{\mathsf{bool}}{\log \circ\, \mathsf{n}})
    \end{align*}
  \end{minipage}
  \begin{minipage}{0.45\textwidth}
    \begin{align*}
      \textit{sig}&\\
      \quad \mathsf{G} &: \mathsf{Type}\\
      \quad \mathsf{n} &: G \to \Nat\\
      \quad \mathsf{insertEdge} &: \mathsf{edge} \to (\qpi{G}{G}{\mathsf{n}})\\
      \quad \mathsf{isEdge} &: \mathsf{edge} \to (\qpi{G}{\mathsf{bool}}{\lambda \_.\, 1})
    \end{align*}
  \end{minipage}
  \caption{Cost signatures; left is Alice and right is Bob.
  For simplicity, suppose the vertices are natural numbers and define edges as ordered pairs (\ie
  $\mathsf{edge} \coloneqq \Nat^2$).}
  \label{fig:cost-sigs}
\end{figure}

\subsection{Analyzing the cost of general recursive functions}\label{sec:general-recursion}

Most efficient algorithms are not defined by structural induction on the input
--- their efficiency is the result of exploiting the structure of the data in
clever, nonobvious ways that nevertheless terminate. It is not surprising that
this often cannot be surmised by syntactic means and requires proof. Hence a
type theoretic framework for cost analysis must provide a story for encoding
general recursive algorithms such that the resulting analysis reflects the
expected complexity (and not, for instance, the complexity of the termination
proof).

A well-known and versatile solution to the encoding of general recursive
functions in total type theory is the celebrated Bove--Capretta method~\citep{bove-capretta:2005}.  Any general recursive
program gives rise to an \emph{accessibility predicate} that tracks the pattern
of recursive calls; this accessibility predicate can be glued onto the original
program as a termination metric, and the final (total) function is defined by
proving that every input is accessible.

Cost recurrences provide an alternative to accessibility predicates. The idea is to
parameterize a given program in a \emph{clock}, induced by the cost recurrence,
which can then serve as a termination metric that frees the program to make
whatever recursive calls are required.  This strategy is attractive in the
quantitative setting precisely because cost analysis computes the desired
instantiation of the clock with no additional effort. In contrast the same
method in a framework for pure behavioral properties becomes a technical device
for definition that does not provide further insight into the defined program.
As observed in \citet{niu-harper:2020}, the cost-aware setting
evinces a synergetic relationship between cost analysis itself and
programming with general recursion that is further amplified
in \calf{}: cost structure enables one to effectively encode general recursion,
and general recursion gives rise to programs with interesting cost structure.

\subsection{Related work}

\subsubsection{Recurrence extraction through CBPV}

The CBPV decomposition of cost structure in \calf{} is directly inspired by
recent work on recurrence extraction~\cite{kavvos-morehouse-licata-danner:2019} for functional programs.
In that setting a source language such as CBV PCF is interpreted via a cost-preserving translation into
CBPV, from which a \emph{syntactic} recurrence relation is extracted; the syntactic recurrence is then
translated into a semantic recurrence in a domain appropriate for mathematical manipulation used in algorithm analysis.
The focus of this work is the formalization of the extract-and-solve paradigm used informally in algorithm analysis and
the modularity with respect to the source language afforded by the CBPV decomposition.

Because the extraction process is stratified over different languages, the \emph{bounding theorem} --- the fact that a source program satisfies
a syntactic recurrence relation --- is an external fact.
In contrast \calf{} collapses the distinction between syntactic and semantic recurrence and
is able to express the source program and the cost recurrence in the same
language. Moreover, \calf{} furnishes a rich specification language that allows
us to prove internally that a program is bounded by a given cost.

Another difference is the presence of general recursion in the work on
recurrence extraction. Because we propose \calf{} as a logical framework for
internal reasoning, inclusion of unrestricted fixed-points is a nonstarter.
This does not, however, prevent us from analyzing the cost of general recursive
programs: as discussed in \cref{sec:general-recursion}, a cost bound provides a
termination metric. Of course, some instances of general recursion is not
reducible in this way, such as nontermination; we expect that non-terminating
programs can also be handled by means of a monad for partiality as in the work
of \citet{capretta:2005}.

\subsubsection{Effects in dependent type theory}

The key ingredient that endows \calf{} with enough structure to serve as a
logic for internal reasoning is the integration of dependent types in an
effectful language. We essentially extend the universe-free fragment of the
\dcbpv{} calculus of \citet{pedrot-tabareau:2020} by axioms for the extensional
modality.  The weaning translation of \dcbpv{} is the closest counterpart to
the model we use to prove the soundness of \calf{}.
To define the weaning translation, \citet{pedrot-tabareau:2020}
introduce the concept of the self-algebraic proto-monad, which provides the structure
needed to model computation universes.  Because we do not axiomatize universes,
computation types in \calf{} are interpreted as algebras over a strong monad as
in the usual Eilenberg--Moore models of CBPV. To include universes, we expect
that the \dcbpv{} approach can be further adapted to \calf{} without significant
modification.

\subsubsection{Transparent \emph{vs.} abstract axiomatization of cost structure}\label{sec:trans-abs}

Semantically, free computations $\F{A}$ can be modeled as
free algebras over the monad $\costty \times -$.  We are however careful to
not commit to this fact \emph{internally}; by keeping the type $\F{A}$
abstract, \calf{} ensures that programs cannot drop costs or branch based on
the cost component of their input.

As an example, a language that does not satisfy this noninterference property is the
language of syntactic recurrences ``PCF with costs'' employed
by \citet{kavvos-morehouse-licata-danner:2019}. Indeed, by interpreting $\F{A}$ as $\costty \times
A$, the language of syntactic recurrences is made transparent enough that
programs can spuriously use the cost of an input to choose an output.  However, this is not an issue in that setting because of the
stratification of the source language (of programs) and target language (of
recurrences): such an exotic program lies outside the image of the
interpretation.

\subsubsection{Intensionality in logic and type theory}

\paragraph{Intensional constructs in computational type theory} %

Cost structure in \calf{} aims to capture an intensional aspect of programs,
historically a difficult phenomenon to study type theoretically. Researchers
in the \nuprl{} tradition have made a number of forays into
intensionality beginning with the \plcv{3} language, which included an
operator $\mathsf{isap}$ that distinguished function applications from other
terms~\cite{constable-zlatin:1984}. \citet{constable-crary:2002} later on
introduce a version of type theory equipped with a more restricted form of
intensionality by internalizing parts of the operational semantics, which can
be construed as a form of reflective deep embedding.

\paragraph{Necessity modalities for intensionality}

In the tradition of structural proof theory and modal type theory, the
\emph{necessity} modality $\Box{A}$ has been argued to capture the formal
aspects of staged computation~\citep{davies:1999}. From this perspective,
$\Box{A}$ is the type of \emph{codes} for terms of type $A$.  A detailed
investigation of this folklore was carried out by \citet{kavvos:2017:thesis},
introducing the intensional PCF (iPCF) programming language with an
intensional fixed-point operator whose type is the {G}\"{o}del-L\"{o}b axiom.
Unfortunately, \citeauthor{kavvos:2017}'s investigation revealed that truly
intensional operations such as $\mathsf{isap}$ must be limited to syntactically
closed terms; such a side condition casts doubt on the type theoretic nature of
intensional operations.
In the context of modal type theory, \citet{pfenning:2001} investigates
intensionality through a judgmental distinction between intensional
expressions, extensional terms, and irrelevant proofs.  However, the
internalization of this new judgmental structure as modal operators is not
fully worked out.

A common theme in prior work that aims to capture intensionality within type
theory is that equations are \emph{removed} underneath certain constructors,
consequently refuting most congruence rules and obstructing
presentations by generators and equations. Although tenable for simple theories,
this approach greatly complicates the integration of type dependency, where
congruence rules play a very important role in usability. In the design of
\calf{} we take the complementary perspective of \emph{conditional
extensionality}, where equations expressing extensional/behavioral properties
are \emph{added} in certain contexts. By modeling intension/extension
as another phase distinction, we give an elegant mathematical account of the
intensional content of programs without sacrificing extensionality principles
or speaking of ``equalities'' that do not always hold.

\subsubsection{Type systems for cost analysis}

\paragraph{Linear type systems}
Many current type-theoretic approaches to cost analysis
rely on the notion of \emph{linearity}/non-duplicability of resources.
A prototypical example is \citeauthor{hofmann:2000}'s type system
for programming in bounded space in which heap resources are abstracted into a type $\Diamond$
that is required to construct heap-allocated data structures.
This idea essentially started the line of work in automated amortized resource analysis (AARA) that includes
automatic heap-space bounds \citep{hofmann-jost:2003},
analysis of higher-order programs \citep{jost-hammond-loidl-hofmann:2010},
and a resource-aware version of OCaml (\raml{}) \citep{hoffmann-aehlig-hofmann:2012}.

In these type systems a derivation may be viewed as a \emph{stateful} transformation of the context (\eg consumable resource) into a computation
that satisfies a cost bound, constituting a type-theoretic formulation of \emph{amortized analysis} \citep{tarjan:1985}.
Consequently, a linear/affine treatment of resources is critical for ensuring the soundness of quantitative reasoning,
which states that the derived cost bound suffices for the actual cost as given by a standard cost dynamics.

Although automated systems such as \raml{} are without a doubt extremely useful, their scope is limited
intentionally to maximize automation. Traditionally, cost bounds in \raml{} and its cousins range over multivariate polynomials,
although support for exponential \citep{kahn-hoffmann:2020} and
logarithmic \citep{DBLP:journals/corr/abs-2101-12029} potential functions
have also been added recently.

In contrast there are no restrictions on what cost
bounds can be in \calf{} aside from one's willingness to state and prove them.
It would be of both theoretical and practical interest to integrate an
automated system such as \raml{} with \calf{} so one can manually verify
complex cost bounds and dispatch the automated tool on easier proof
obligations. The difficulty here is integrating the \raml{}'s affine type
system with \calf{}'s dependent type theory. We expect recent work on
substructural logical frameworks and linear dependent types to have some
bearing on this problem
\citep{licata-shulman-riley:2017,atkey:2018,krishnaswami-pradic-benton:2015}.

\paragraph{Non-amortized cost analysis}
Type theoretic formulations of cost analysis do not have to be based on amortization:
\citet{crary-weirich:2000} develop a type system
for resource bound certification by means of a \emph{virtual clock}.
Function types are refined with a starting and ending time, so that
a function of type $(A,5) \to (B,0)$ is an ordinary function $A \to B$ with the property that
it is to be applied when the clock is five and completes when the clock is zero.
Clock polymorphism relaxes the limit on the starting time by
allowing one to form the type $\forall n.\, (A,n+5) \to (B,n)$.
Variable cost bounds are definable via a limited form of dependency
using inductive kinds, which unfortunately imposes a somewhat stilted programming style.

More recently, \citet{wang-wang-chlipala:2017} introduced TiML,
a language loosely based on Standard ML that provides internal cost specifications
in the form of a timed function type.
TiML supports indexed data types whose indices furnish a notion of size measure,
leading to a more natural treatment of variable cost bounds compared to \citet{crary-weirich:2000}.
The TiML type system generates verification conditions that are further refined by a recurrence solver
using heuristics such as the Master Theorem~\citep{cormen-leiserson-rivest-stein:2009}.

Type systems presented in \citet{wang-wang-chlipala:2017, crary-weirich:2000} represent practical compromises in the sense that they are
primarily designed for expressing cost information and only secondarily support limited forms of behavioral specification.
In contrast \calf{} is a full-spectrum dependent type theory designed for \emph{both} quantitative and behavioral verification.

\paragraph{Frameworks for cost analysis}
\calf{} is a \emph{framework} for cost analysis in the sense that
it provides the language for speaking about the cost structure of programs but does not
prescribe a particular method for cost analysis.
Recently, \citet{rajani-gaboardi-garg-hoffmann:2021} advance a similar thesis by developing a type theory, $\lambda$-amor,
that unifies many extant type systems for cost analysis, in particular exhibiting $\lambda$-amor embeddings of both
effect and coeffect-based systems for cost accounting.
However, because $\lambda$-amor does not support dependent types,
there is no satisfying account of the behavioral fragment.

In the context of Liquid Haskell, \citet{handley-vazou-hutton:2019} define a monadic library called RTick
for reasoning about both quantitative and correctness properties by taking advantage of Liquid Haskell's refinement
type system. They substantiate the library with a rich repository of examples, including sorting algorithms,
programs optimizations, and relational cost analysis. However, the RTick library
suffers from the usual problem of representing cost structure transparently (see \cref{sec:trans-abs}); consequently
there is no guarantee that programs actually accumulate cost as intended.

Finally we mention the work of \citet{niu-harper:2020} on a cost-aware
computational type theory \catt{} in the \nuprl{} tradition. Unlike the type
theory of \citet{constable-crary:2002}, \catt{} only internalizes cost
structure, which leads to a framework that is more directly applicable to cost
analysis.
In particular, \citet{niu-harper:2020} introduce a connective ``funtime''
that internalizes cost specification on functions types and prove
a novel refinement rule for funtime by appealing to the specified cost bound,
constituting an induction principle based on cost structure.
Our observation that cost analysis may be used to encode general recursion in \calf{} is
inspired by the work on \catt{}, as is the idea of using a cost-aware dependent function type
to specify cost signatures. \citet{niu-harper:2020} do not develop a formal proof theory for \catt{},
a fact that appears to pose significant challenges for its mechanization.

\subsubsection{Separation logic and Isabelle/HOL}

An alternative perspective, exemplified by the work of \citet{atkey:2010} on
amortized resource analysis in separation logic, is to treat cost as an
ownable resource. Program logics in this tradition primarily focus on the
verification of imperative programs.  \citeauthor{atkey:2010}'s formulation
essentially transposes the types-with-potential concept of
\citet{hofmann-jost:2003} into the imperative setting, allowing one to prove
resource bounds on heap-based data structures.

More recently, \citet{mevel-jourdan-pottier:2019} employed similar ideas to develop a resource-aware extension to the Iris program logic~\citep{iris:2015,iris:2018}.
The interesting twist in this work is the use of time \emph{receipts}, which are dual to the more common time \emph{credits}.
Time receipts witness that a computation takes at least a certain amount of resources, thereby
establishing a lower bound on the cost of programs. This can be used to prove that catastrophic events do not happen
until a long time has passed. An application of the framework is the verification of an asymptotically tight upper bound
on union-find, a mathematically involved and complex proof.

Iris is a very powerful tool whose scope goes far beyond cost analysis; the
theoretical overhead of Iris when applied \emph{specifically} to quantitative
analysis of functional programs is consequently somewhat high in contrast to
the basic rules of \calf{} which can be written down in half a page.
Furthermore, the intended semantics of \calf{} can be interpreted somewhat
simple-mindedly in \emph{any} topos equipped with a subterminal sheaf
representing the partition between extension and intension.
In this respect, \calf{} offers a
fundamentally different perspective on cost analysis based on the
\emph{synthetic} integration of cost specification into a full-spectrum
dependent type theory rather than the definition of a resource-sensitive
program logic over an existing language.

\subsubsection{Isabelle/HOL}

The proof assistant Isabelle/HOL represents another hot spot for complexity verification.
In this setting the \emph{Archive of Formal Proofs} contains a number of case studies on complexity verification, including
quicksort~\citep{eberl:2017:quicksort}, medians of medians~\citep{eberl:2017:median-of-medians}, and
the formalization of the Akra-Bazzi theorem~\citep{eberl:2015:akra-bazzi}, just to name a few.
In more recent work \citet{functional-algorithms-verified} give a systematic study of the functional correctness and complexity verification of
a variety of algorithms and techniques including sorting, search trees, amortized analysis, dynamic programming, \etc

In these works cost is often instrumented through the writer monad $\Nat \times
-$ or just treated informally. In contrast, \calf{} allows the user to define
formal relations between programs and recurrences, and the careful
instrumentation of the stepping effect induces a noninterference property not
found in the Isabelle/HOL setting.  The Isabelle/HOL approach to cost analysis
uses existing tools in the framework to \emph{encode} the notion of cost, while
\calf{} is a framework in which one can use type-theoretic principles to reason
about cost/quantitative properties of programs in a first class way without
sacrificing the connection to the uninstrumented programs.

\section{Cost-aware logical framework}

We define \calf{} as an extension to the \dcbpv{} calculus of
\citet{pedrot-tabareau:2020}.  As discussed in \cref{sec:intro}, the design of
\calf{} rests on three main pillars.  First, the fine-grained type structure of
CBPV gives a compositional account of cost at higher types. Secondly, in the
dependent setting \dcbpv{} provides a smooth integration of effects and type
dependency, which allows us to define cost-aware programs and prove theorems
about them in a single language.  Lastly (as in
\cref{sec:lang-of-phase-distinctions}), the extensional phase $\ExtOpn$
generates a pair of complementary open and closed modalities $\Op,\Cl$ in the
sense of \citet{rijke-shulman-spitters:2017,schultz-spivak:2019} that govern the interaction
between \emph{intension} and \emph{extension}.  In the following, we introduce
\calf{} at an informal level through simple examples that illustrate the cost
effect $\mathsf{step}$, internal cost bounds as equations, and the interplay of
the \emph{open/extensional modality} $\Op$ and the \emph{closed/intensional
modality} $\Cl$.

\subsection{A refresher on CBPV: the identity function two ways}

We give a quick introduction to CBPV through the simplest possible example: the
identity function (on natural numbers).  Recall that the type structure of CBPV
is centered around the polarization of values and computations.  For our
example, consider the following selection of types and terms:\footnote{\calf{}
also includes additional types such as dependent products and dependent sums.}

\begin{center}
\begin{tabular}{ l l }
\textbf{Values} & \textbf{Computations}\\
$A,B \coloneqq \UU{X}, \nat$ & $X,Y \coloneqq \F{A}, A \to X$\\
$a,b \coloneqq \thunk{e}, \zero, \suc{a}$ & $e,f \coloneqq \ret{a}, \bind{e}{f}, \force{u}, \reccst{a}{e_1}{e_2}, \lambda a.\, e, \ap{f}{v}$
\end{tabular}
\end{center}

The pair of type constructors $\mathsf{F}$ and $\mathsf{U}$ bridges the dichotomy between value and computation types:
$\mathsf{F}$ turns a value \isof{a}{A} into the computation \isof{\ret{a}}{\F{A}}, and
$\mathsf{U}$ reifies a computation \isof{e}{X} into a value \isof{\thunk{e}}{\UU{X}}.
Observe that functions are \emph{computations} in CBPV, a phenomenon that may be
explained by examining the operational behavior of functions in the CK-machine model of CBPV \citep{levy:2006}.
The fine-grained type structure of CBPV evinces embeddings of both CBV and CBN.
For instance, one may recover the CBV function space $\nat \to_{\mathsf{cbv}} \nat$ as the
CBPV type $\UU{\nat \to \F{\nat}}$. We refer the reader to \citet{levy:2004} for a more thorough introduction.

For the purposes of our example, we only consider the value type $\nat$. As usual, \zero{} and \suc{n} are values of $\nat$.
The (nondependent) recursor on $\nat$ has the following type:
\begin{align*}
\mathsf{rec} &: \impl{X} \nat \to X \to (\nat \to \UU{X} \to X) \to X
\end{align*}
Note that the recursive call is reified as a value $\UU{X}$ because variables range over values in CBPV.
If we restrict attention to natural numbers, there are
two evident ways to compute the identity: one program returns the argument
immediately, and the other reconstructs the argument by recursion.
In CBPV, they are rendered as the following programs:

\begin{minipage}{.2\textwidth}
  \begin{align*}
    \textit{id}_{\mathsf{easy}} &: \nat \to \F{\nat}\\
    \textit{id}_{\mathsf{easy}} &= \lambda x.\, \ret{x}
  \end{align*}
  \end{minipage}
  \begin{minipage}{.8\textwidth}
  \begin{align*}
    \textit{id}_{\mathsf{hard}} &: \nat \to \F{\nat}\\
    \textit{id}_{\mathsf{hard}} &= \lambda x.\,
      \reccst{x}{\ret{\zero}}{\lambda x',u.\, \bind{\force{u}}{\lambda y.\, \ret{\suc{y}}}}
\end{align*}
\end{minipage}
\ \\
\indent Note that in $\textit{id}_{\mathsf{hard}}$ we have to force the reified recursive computation \isof{u}{\UU{\F{\nat}}} to
obtain a computation \F{A}, thence sequencing it and tacking on an additional successor.

\subsection{Cost monoid: cost structure of programs}\label{sec:cost-structure}

Cost-aware programs carry quantitative information through elements of the cost monoid $\mathbb{C}$.
Because different algorithms and cost models require different notions of cost,
we parameterize \calf{} by an arbitrary \emph{cancellative monoid} $(\mathbb{C}, +, 0)$;
here cancellative means that the operation $+$ is injective, a property that
is needed to establish metatheoretic results in \cref{sec:metatheory}.
Further structure on $\mathbb{C}$ can be negotiated depending on one's preference for generality.
For the purposes of analyzing (upper) bounds of algorithms, it is reasonable to additionally require the structure of
an \emph{ordered monoid} $(\mathbb{C}, +, 0, \le)$ in which the monoid multiplication is compatible with a  preorder $\le$.

\subsection{Cost as an effect in \calf{}}\label{sec:cost-as-an-effect}

We formulate cost in \calf{} as a primitive effect by adding
a new form of computation \Alert{\mstept{X}{c}{e}} that is parameterized by a computation type $X$ and
an element of the cost monoid $c$.
The meaning of \mstept{X}{c}{e} is to effect $c$ units of cost and continue as $e$;
consequently, we require that $\mathsf{step}$ is coherent with the monoid structure on $\mathbb{C}$:
\begin{mathpar}
  \mstept{X}{0}{e} = e
  \and
  \mstept{X}{c}{\mstept{X}{d}{e}} = \mstept{X}{c+d}{e}
\end{mathpar}

In addition, we require a slew of equations governing the interaction of $\mathsf{step}$ with
other computations. For instance, $\mathsf{step}$ satisfies the following laws:
\begin{align*}
  \mathsf{bind}_{\mathsf{step}} &: \impl{\isof{e}{\F{A}}, \isof{f}{A \to X}} \bind{\mstept{\F{A}}{c}{e}}{f} = \mstept{X}{c}{\bind{e}{f}}\\
  \mathsf{lam}_{\mathsf{step}} &: \mstept{A \to X}{c}{\lambda x.\, e} = \lambda x.\, \mstept{X}{c}{e}
\end{align*}

The first equation states that $\mathsf{step}$ inside a sequence of computations can be commuted
outside and executed first; the second equation states that $\mathsf{step}$ commutes with abstraction.

\paragraph{Meaning of $\mathsf{step}$ in the Eilenberg–Moore model of \calf{}}
Each computation type $X$ of \calf{} is interpreted as an algebra $(\carrier{X}, \alpha)$
over $\mathbb{C} \times \--$. Thus $\mathsf{step}_X$ is interpreted by the structure map $\alpha$,
and all of the equations associated with $\mathsf{step}_X$ hold as a consequence of the algebra laws.

\paragraph{Cost of identity}
For the identity example, let us suppose that $\mathbb{C}$ is the additive monoid on $\Nat$
under the usual ordering.
Consider the two identity programs from the previous section. Suppose that we wanted to charge unit cost
for each recursive call in the program. In \calf{}, we can achieve this by instrumenting the
program with $\mathsf{step}$ at the appropriate place:
\begin{align*}
  \textit{id}_{\mathsf{hard}} &= \lambda x.\,
    \rec{x}{\F{\nat}}{\ret{\zero}}{\lambda x',u.\, \Alert{\mathsf{step}_{\F{\nat}}^1}(\bind{\force{u}}{\lambda y.\, \ret{\suc{y}}})}
\end{align*}

We do nothing for $\textit{id}_{\mathsf{easy}}$ because there is no recursion involved.

\subsection{Cost refinements in \calf{}}

Recall the predicate $\mathsf{hasCost}$ from \cref{sec:intro},
$\mathsf{hasCost}(A,e,c) \coloneqq \Sigma \isof{a}{A}.\, e =_{\F{A}} \mstep{c}{\ret{a}}$,
which states that the computation \isof{e}{\F{A}} incurs $c$ units of cost.
Given our instrumented identities, we can prove the following quantitative
refinements for \ideasy{} and \idhard{}:

\begin{theorem}[\Mech{Examples.Id.Easy.id≤id/cost}]
  We have that $\ideasy(x)$ has cost 0 for all \isof{x}{\nat}.
\end{theorem}

\begin{proof}
  We take the input as the witness value and apply the coherence rule of $\mathsf{step}$ to obtain
  $\ideasy(x) = \ret{x} = \mstept{X}{0}{\ret{x}}$.
\end{proof}

\begin{theorem}[\Mech{Examples.Id.Hard.id≤id/cost/closed}]
  We have that $\idhard(x)$ has cost $\iota(x)$ for all \isof{x}{\nat}, where
  $\iota$ is the obvious monoid isomorphism $\nat \cong \Nat$.
\end{theorem}

\begin{proof}
We proceed by induction on $x$. In the inductive case, we use the
equations governing $\mathsf{step}$ explained in \cref{sec:cost-as-an-effect}
and the inversion principles for $\mathsf{U}$ and $\mathsf{F}$.
\end{proof}

The study of quantitative properties \emph{qua} equations evinces the essential advantage of verification in \calf{}:
proof of quantitative properties is reduced to ordinary equational reasoning.

\subsection{Reasoning about extensional properties using \texorpdfstring{$\ExtOpn$}{the extensional phase}}

In general, equations between cost-aware programs of \calf{} are in some sense rare, exactly
because the cost effect obstructs equations between extensionally equivalent computations.
To account for extensional equivalence and other behavioral properties, we study programs
in the fragment of \calf{} under the extensional phase $\ExtOpn$.

\paragraph{The extensional fragment of \calf{}}
As discussed in \cref{sec:lang-of-phase-distinctions}, the (proof-irrelevant) proposition $\ExtOpn$
renders the extensional modality as the function space $\Op A \coloneqq \ExtOpn \to A$,
which naturally generalizes to a dependent modality $\Op_{\isof{u}{\ExtOpn}} (A(u)) \coloneqq (u : \ExtOpn) \to A(u)$.
The force of this modality is effected by the following axiom in \calf{}, which
makes $\mathsf{step}$ silent in the presence of $\ExtOpn$:
\begin{align*}
  \Alert{\mathsf{step}/{\ExtOpn}} &: \Op(\mstept{X}{c}{e} = e)
\end{align*}

Thus the extensional modality $\Op$ governs \emph{behavioral} specifications
in the sense that any type in the image of $\Op$ is oblivious to computation steps.
One such behavioral specification is the \emph{extensional equality} between programs,
rendered in \calf{} as the type $\Op(e_1 = e_2)$.
In the case of the two identity programs, we can take \ideasy{} as the specification
and prove that \idhard{} obeys it:

\begin{theorem}[\Mech{Examples.Id.easy≡hard}]
  We have the modal equation $\Op(\idhard = \ideasy)$.
\end{theorem}

\begin{proof}
  We give a sketch of the proof. By standard results~\citep{rijke-shulman-spitters:2017} on the extensional modality
  generated by a proposition, we have that $\Op$ is a left exact, monadic modality. Thus it suffices to show that
  $\eta_{\Op}(\idhard) =_{\Op(\nat \to \F{\nat})} \eta_{\Op}(\ideasy)$, where
  $\eta_{\Op}(x) \coloneqq \lambda \isof{u}{\ExtOpn}.\, x$ is the monadic unit for $\Op$.
  By function extensionality, it suffices to show $\idhard(x) = \ideasy(x)$ for all \isof{x}{\nat} and \isof{u}{\ExtOpn}.
  This follows by induction on $x$, using the equation $\Alert{\mathsf{step}/{\ExtOpn}(u)}$ in the inductive case.
\end{proof}

\subsection{Closed/intensional modality}\label{sec:closed-modality}

Complementary to the open/extensional modality $\Op$ is the \emph{closed/intensional modality} $\Cl$ that governs
intensional/quantitative properties.
One may think of applying the intensional modality as sealing away the extensional part so it cannot be observed
and leaving only the intensional part of a program.
Consequently a type in the image of the intensional modality $\Cl$ is trivial under the extensional modality $\Op$, \ie
$\Op{\Cl{A}} \cong 1$ for any type $A$.
In the Eilenberg–Moore model of \calf{}, we exploit this property to enforce the $\mathsf{step}$ erasure rule
$\mathsf{step}/\ExtOpn$: by interpreting computation types as algebras for the writer monad $\Cl{\mathbb{C}} \times \--$,
the cost structure of programs is obliterated whenever the extensional phase $\ExtOpn$ is present in the context.
In \cref{fig:closed} we define the intensional modality as a quotient inductive type~\citep{altenkirch-kaposi:2016,fiore-pitts-steenkamp:2021}.
In categorical language the intensional modality is the pushout $A \sqcup_{A \times \ExtOpn} \ExtOpn$
of the projection maps of $A \times \ExtOpn$.

\begin{footnotesize}
\begin{figure}
  \begin{minipage}[b]{0.3\textwidth}
    \begin{align*}
      \textbf{data}&\;\Cl{A}\; \textbf{where}\\
      \eta_{\Cl} &: A \to \Cl{A}\\
      \ast &: \ExtOpn \to \Cl{A}\\
      \_ &: \Pi \isof{a}{A}.\, \Pi \isof{u}{\ExtOpn}.\\
      &\eta_{\Cl}(a) = \ast(u)
    \end{align*}
  \end{minipage}%
  \begin{minipage}[b]{0.7\textwidth}
      \begin{align*}
        & \Op{(A^{\Downarrow c} \to B^{\Downarrow d})} \cong \Op{A^{\Downarrow c}} \to \Op{B^{\Downarrow d}} \\
        &= \Op{(\Sigma \isof{e}{A}.\, \Cl{\mathsf{hasCost}(A,e,c)})} \to
       \Op{(\Sigma \isof{f}{B}.\, \Cl{\mathsf{hasCost}(B,f,d)})} \\
        &\cong \Sigma \isof{e}{\Op{A}}.\, \Op_{\isof{u}{\ExtOpn}}{\Cl{\mathsf{hasCost}(A,e(u),c)}} \to
        \Sigma \isof{f}{\Op{B}}.\, \Op_{\isof{u}{\ExtOpn}}{\Cl{\mathsf{hasCost}(B,f(u),d)}}\\
        &\cong \Sigma \isof{e}{\Op{A}}.\, 1 \to \Sigma \isof{f}{\Op{B}}.\, 1\\
        &\cong \Op{A} \to \Op{B} \cong \Op(A \to B)
      \end{align*}
  \end{minipage}
  \caption{Left: closed modality as a quotient inductive type; right: extracting the extensional content of cost-aware functions, where $\Op_{\isof{u}{\ExtOpn}} A(u) \coloneqq (\isof{u}{\ExtOpn}) \to A(u)$
  is the dependent version of the extensional modality.}
  \label{fig:closed}
\end{figure}
\end{footnotesize}

\paragraph{Program extraction}

The intensional modality allows one to organize quantitative information in a way that facilitates extraction of
ordinary programs from cost-aware programs.
For instance, consider the type of functions between cost-aware computations
$A^{\Downarrow c} \to B^{\Downarrow d}$
where $A^{\Downarrow c} \coloneqq \Sigma \isof{e}{A}.\, \Cl{\mathsf{hasCost}(A,e,c)}$ is the type of
computations of $A$ that consume $c$ steps. Note that the type $\mathsf{hasCost}$ is guarded by the intensional modality.
\cref{fig:closed} shows that one may extract the underlying function by applying the \emph{extensional modality} and
using the fact that it is lex and commutes with exponentials.

\subsection{Noninterference}

In \cref{sec:intro} we claim that the extensional part of \calf{} programs cannot analyze the
intensional part of their input. This is substantiated by the following theorem:
\begin{theorem}[\Mech{Calf.Noninterference.oblivious}]\label{thm:step-oblivious}
  Given any function \isof{f}{\F{A} \to \Op{B}}, we have that
  $f(\mstep{c}{e}) = f(e)$ for any \isof{c}{\mathbb{C}} and \isof{e}{\F{A}}.
\end{theorem}

Moreover, when the input of a \calf{} program is fully intensional, we may obtain a stronger
noninterference property by observing the interaction of the extensional
and intensional modalities (also exploited in the program extraction example in \cref{sec:closed-modality}):

\begin{theorem}[\Mech{Calf.Noninterference.constant}]\label{thm:non-interference}
  Any function \isof{f}{\Cl{A} \to \Op{B}} is constant.
\end{theorem}

In other words one cannot construct a map into the extensional fragment that branches based on purely intensional information,
a fact that enables a type-directed method to systematically eliminate intensional structure from
programs of a certain shape. For instance, one may perform the following \emph{program optimization}:

\begin{theorem}[\Mech{Calf.Noninterference.optimization}]\label{thm:optimization}
  Any map \isof{f}{(\Sigma \isof{c}{C}.\, \Cl{A(c)}) \to \Op{B}} admits an optimization
  \isof{f'}{C \to \Op{B}} such that $f(c,a) = f'(c)$ for all \isof{c}{C} and \isof{a}{\Cl{A(c)}}.
\end{theorem}

\begin{proof}
  We need to construct a map \isof{f'}{C \to \Op{B}}. Suppose that \isof{c}{C}. By \cref{thm:non-interference},
  there is a constant map \isof{\lambda \_.\, b}{\Cl{A(c)} \to \Op{B}} such that
  $\lambda a.\, f(c,a) = \lambda \_.\, b$, which provides the required program $b$ of type $\Op{B}$.
  By definition, we have $f(c,a) = f'(c)$ for all \isof{c}{C} and \isof{a}{\Cl{A(c)}}.
\end{proof}

Recalling the type of bounded computations $A^{\Downarrow c}$ from \cref{sec:closed-modality},
we may apply \cref{thm:optimization} to a program of type $A^{\Downarrow c} \to \Op{B}$
to obtain an optimized program of type $A \to \Op{B}$ that dispenses with the proof of the cost bound.

\subsection{Presentation of \calf{} in a logical framework}\label{sec:calf}

Following recent work~\citep{uemura:2019, sterling-harper:2021, gratzer-sterling:2020} promoting the study of
type theories \emph{qua} mathematical objects in structured categories, we present \calf{} as a signature
in a \emph{logical framework} using the internal language of locally cartesian closed categories (lcccs).
As observed by~\citet{uemura:2019}, one can specify a type theory
as a list of constants in a version of extensional dependent type theory.
The resulting signature \emph{presents} the free lccc over the defined constants, which we then take
as the \emph{definition} of the type theory. As we show in \cref{sec:metatheory}, this view of
type theories as certain initial objects allows one to easily define models of \calf{}.

Concretely, we work in a logical framework with a universe of judgments \jdg{} closed under dependent
product, dependent sum, and extensional equality. An object theory (\eg \calf{}) is specified as follows:
\begin{enumerate}
  \item Judgments are declared as constants ending in \jdg{}.
  \item Binding and scope is handled by the framework-level dependent product $(\isof{x}{X}) \to Y(x)$.
  \item Equations between object-level terms are specified by constants ending in the framework-level
  equality type $x_1 =_X x_2$.
\end{enumerate}

\paragraph{Presentation of \calf{} in the logical framework}
We present \calf{} in \cref{fig:calf}. For brevity,
we do not explicitly mention all types and computations here, the majority of which remain unchanged from \dcbpv{};
however the full definition may be found in \cref{sec:calf-def}.
Note that we define computations as $\tmc{X} \coloneqq \tmv{\UU{X}}$, leading to a less bureaucratic version of CBPV
in which thunk and force are identities.
Observe that the \calf{} equality type $\mathsf{eq}$ comes equipped with a reflection
rule $\mathsf{ref}$ that renders inhabitation of $\mathsf{eq}$ equi-derivable with judgmental equality.
Thus we abuse notation slightly and also write $e =_{\tmc{X}} f$ for the type $\eqty{\UU{X}}{e}{f}$.

\begin{small}
\begin{figure}
  \begin{minipage}[t]{0.45\textwidth}
   \begin{align*}
    \mathbb{C} &: \jdg\\
    0 &: \mathbb{C}\\
    + &: \mathbb{C} \to \mathbb{C} \to \mathbb{C}\\
    \le &: \mathbb{C} \to \mathbb{C} \to \jdg{}\\
    \mathsf{costMon} &: \mathsf{isCostMonoid}(\mathbb{C}, 0, +, \le)\\
    \mathsf{step} &: \impl{\isof{X}{\tpc}} \mathbb{C} \to \tmc{X} \to \tmc{X}\\
    \mathsf{step}_{0} &: \impl{X,e} \mstep{0}{e} = e\\
    \mathsf{step}_{+} &: \impl{X,e,c_1,c_2}\\
    &\mstep{c_1}{\mstep{c_2}{e}} = \mstep{c_1 + c_2}{e}
  \end{align*}
  \end{minipage}
  \begin{minipage}[t]{0.45\textwidth}
    \begin{align*}
    \tpv &: \jdg\\
    \mathsf{tm}^+ &: \tpv \to \jdg\\
    \mathsf{U} &: \tpc \to \tpv\\
    \mathsf{F} &: \tpv \to \tpc\\
    \mathsf{tm}^{\ominus}(X) &\coloneqq \tmv{\UU{X}}\\
    \mathsf{ret} &: (\isof{A}{\tpv}, \isof{a}{\tmv{A}}) \to \tmc{\F{A}}\\
    \mathsf{bind} &: \impl{\isof{A}{\tpv}, \isof{X}{\tpc}} \tmc{\F{A}} \to\\
    &(\tmv{A} \to \tmc{X}) \to \tmc{X}
    \end{align*}
  \end{minipage}
  \begin{minipage}{0.45\textwidth}
    \begin{align*}
    \ExtOpn &: \jdg\\
    \ExtOpn/{\mathsf{uni}} &: \impl{\isof{u,v}{\ExtOpn}} u = v\\\\
    \ext{\mathcal{J}} &\coloneqq \ExtOpn \to \mathcal{J}\\
    \mathsf{step}/{\ExtOpn} &: \impl{X, e, c} \ext{\mstep{c}{e} = e}\\
    \Op^+ &: \tpv \to \tpv\\
    \_ &: \impl{A} \tmv{\Op^+{A}} \cong \Op(\tmv{A})
    \end{align*}
  \end{minipage}
  \begin{minipage}{0.45\textwidth}
    \begin{align*}
    \Cl &: \tpv \to \tpv\\
    \eta_{\Cl} &: \tmv{A} \to \tmv{\Cl{A}}\\
    \ast &: \ExtOpn \to \tmv{\Cl{A}}\\
    \_ &: \Pi \isof{a}{\tmv{A}}.\, \Pi \isof{u}{\ExtOpn}.\, \eta_{\Cl}(a) = \ast(u)\\
    \mathsf{ind}_{\Cl} &: \impl{A} (\isof{a}{\tmv{\Cl{A}}}) \to (\isof{X}{\tmv{\Cl{A}} \to \tpc}) \to\\
    &(\isof{x_0}{(\isof{a}{\tmv{A}}) \to \tmc{X(\eta_{\Cl}(a))}}) \to\\
    &(\isof{x_1}{(\isof{u}{\ExtOpn}) \to \tmc{X(\ast(u))}}) \to\\
    &((\isof{a}{\tmv{A}}) \to (\isof{u}{\ExtOpn}) \to x_0(a) = x_1(u)) \to\\
    &\tmc{X(a)}
    \end{align*}
  \end{minipage}
  \begin{minipage}{\textwidth}
    \begin{align*}
     \Pi &: (\isof{A}{\tpv}, \isof{X}{\tmv{A} \to \tpc}) \to \tpc\\
    (\mathsf{ap}, \mathsf{lam}) &: \impl{A,X} \tmc{\Pi(A; X)} \cong (\isof{a}{\tmv{A}}) \to \tmc{X(a)}\\
    \Sigma^{++} &: (\isof{A}{\tpv}, \isof{B}{\tmv{A} \to \tpv}) \to \tpv\\
    (\mathsf{unpair}^{++}, \mathsf{pair}^{++}) &: \impl{A,B} \tmv{\sigpos{A}{B}} \cong \Sigma (\tmv{A}) (\lambda a.\, \tmv{B(a)})\\
    \Sigma^{+-} &: (\isof{A}{\tpv}, \isof{X}{\tmv{A} \to \tpc}) \to \tpc\\
    (\mathsf{unpair}^{+-}, \mathsf{pair}^{+-}) &: \impl{A,X} \tmc{\signeg{A}{X}} \cong \Sigma (\tmv{A}) (\lambda a.\, \tmc{X(a)})
    \end{align*}
  \end{minipage}

  \begin{minipage}{0.5\textwidth}
    \begin{align*}
    \mathsf{eq} &: (\isof{A}{\tpv}) \to \tmv{A} \to \tmv{A} \to \tpv\\
    \mathsf{self} &: \impl{A} (\isof{a,b}{\tmv{A}}) \to\\
    & a =_{\tmv{A}} b \to \tmv{\eqty{A}{a}{b}}\\
    \mathsf{ref} &: \impl{A} (\isof{a,b}{\tmv{A}}) \to\\
    &\tmc{\F{\eqty{A}{a}{b}}} \to a =_{\tmv{A}} b\\
    \mathsf{uni} &: \impl{A,a,b} (\isof{p,q}{\tmc{\F{\eqty{A}{a}{b}}}}) \to \Op{(p = q)}
    \end{align*}
  \end{minipage}%
  \begin{minipage}{0.5\textwidth}
    \begin{align*}
      \nat &: \tpv \\
      \zero &: \tmv{\nat} \\
      \mathsf{suc} &: \tmv{\nat} \to \tmv{\nat}\\
      \mathsf{rec} &: (\isof{n}{\tmv{\nat}}) \to\\
      &(\isof{X}{\tmv{\nat} \to \tpc}) \to \tmc{X(\zero)} \to\\
      &((\isof{n}{\tmv{\nat}}) \to \tmc{X(n)} \to\\
      &\quad \tmc{X(\suc{n})}) \to \tmc{X(n)}
    \end{align*}
  \end{minipage}

  \begin{minipage}{\textwidth}
    \begin{align*}
    \mathsf{lam}_{\mathsf{step}} &: \impl{A,X,f,c} \lam{\mstep{c}{f}} = \mstep{c}{\lam{f}}\\
    \mathsf{pair}_{\mathsf{step}} &: \impl{A,X,e_1,e_2,c} \mstep{c}{(e_1,e_2)} = (e_1 , \mstep{c}{e_2})\\
    \mathsf{bind}_{\mathsf{step}} &: \impl{A,X,e,f,c} \bind{\mstep{c}{e}}{f} = \mstep{c}{\bind{e}{f}}
    \end{align*}
  \end{minipage}

  \caption{Equational presentation of \calf{} as a signature $\Sigma_{\calf}$ in the logical framework. Here
  the type $\mathsf{isCostMonoid}$ encodes all the structure of a cost monoid and $\Sigma$ denotes the framework-level
  dependent sum. We write $(\alpha,\beta) : A \cong B$ when $\alpha$ and $\beta$ are the forward map and backward map
  of an isomorphism $A \cong B$.}
  \label{fig:calf}
\end{figure}
\end{small}

\section{Quantitative refinement in \calf{}}\label{sec:quantitative-refinements}

In \cref{sec:calf} we developed the skeletal structure of \calf{} equipped the effect $\mathsf{step}$ for cost instrumentation.
In this section, we define a quantitative refinement expressing the \emph{upper bound} of a computation and present a collection of expected rules for the refinement relation.
As mentioned in \cref{sec:intro}, this mode of quantitative reasoning manifests as \emph{equations} between computations;
we can make meaningful inferences about the cost of a computation $e$ by equating it to another computation whose cost structure is readily available,
\ie $\mstep{c}{\ret{a}}$.

As a first attempt, we may conjecture that a computation \isof{e}{\tmc{\F{A}}} is bounded
by \isof{c}{\mathbb{C}} if $e =_{\tmc{\F{A}}} \mstep{c'}{\ret{a}}$ for some $c' \le c$ and \isof{a}{\tmv{A}}.
While this is a perfectly sensible definition, our investigations suggest it is more natural to replace ordinary
inequality $\le$ with the \emph{extensional inequality} $\Op{(c' \le c)}$.
Consequently the upper bound specification may be concisely expressed
by refining the $\mathsf{hasCost}$ refinement
from \cref{sec:quant-reasoning}:
\begin{align*}
  \mathsf{hasCost}(A, e, c) &= \Sigma^{++} \isof{a}{A}.\, \eqtycst{\tmc{\F{A}}}{e}{\mstep{c}{\ret{a}}}\\
  \mathsf{isBounded}(A,e,c) &= \Sigma^{++} \isof{c'}{\UU{\widehat{\mathbb{C}}}}.\, \Op^+(\UU{c' \mathrel{\widehat{\le}} c}) \times \mathsf{hasCost}(A,e,c')
\end{align*}

Here $\widehat{\mathbb{C}}$ and $\widehat{\le}$ internalizes the (judgmental) structure of the cost monoid
$\mathbb{C}$ as \calf{} types, \ie we have that $\tmc{\widehat{\mathbb{C}}} \cong \mathbb{C}$;
the full axiomatization of $\widehat{\mathbb{C}}$ and $\widehat{\le}$ may be found in \cref{fig:calf-cost} in \cref{sec:calf-def}.
The use of the extensional inequality in the $\mathsf{isBounded}$ refinement
reflects the intuition that ``costs don't have cost''. More importantly, this arrangement grants one
access to the extensional fragment and the \emph{extensional} properties therein when proving
cost refinements, which is essential for analyses of algorithms that depend on
behavioral invariants of data structures.
In \cref{sec:metatheory}, we prove the validity of cost bounds defined in this manner
by exhibiting an equivalence between ordinary inequality and extensional inequality
for a large class of cost monoids.

\subsection{Quantitative refinement rules}

\calf{} admits many expected principles for reasoning about the $\mathsf{isBounded}$
refinement. We present the exemplary rules used in the case
studies in \cref{sec:verification}, summarized in inference rule style in
\cref{fig:refinement-lemmas}.
There are three syntax-directed refinements:
the \textsc{Return} refinement bounds the return of a value by the neutral element \isof{0}{\mathbb{C}};
the \textsc{Step} refinement states that a $\mathsf{step}$ increases the bound on a given computation by the amount charged;
the \textsc{Bind} refinement combines the bounds on the two constituent computations in the obvious way.
Lastly, the \textsc{Relax} refinement allows a cost bound to be replaced with a weaker bound.
Proofs of all refinement lemmas can be found in the supplementary material.

\def\stackalignment{l}
\begin{figure}
  \begin{mathpar}
  \inferrule[\stackunder{Return}{(\Mech{Calf.Types.Bounded.bound/ret})}]{
  }{
    \isBounded{A}{\ret{a}}{0}
  }

  \inferrule[\stackunder{Step}{(\Mech{Calf.Types.Bounded.bound/step})}]{
    \isBounded{A}{e}{c}
  }{
    \isBounded{A}{\mstep{d}{e}}{d+c}
  }

  \inferrule[\stackunder{Bind}{(\Mech{Calf.Types.Bounded.bound/bind})}]{
    \isBounded{A}{e}{c}\\
    \forall \isof{a}{A}.\, \isBounded{B}{f(a)}{d(a)}
  }{
    \isBounded{B}{\bind{e}{f}}{\bind{e}{\lambda a.\, c + d(a)}}
  }

  \inferrule[\stackunder{Relax}{(\Mech{Calf.Types.Bounded.bound/relax})}]{
    \isBounded{A}{e}{c}\\
    c \le c'
  }{
    \isBounded{A}{e}{c'}
  }
\end{mathpar}
\caption{Quantitative refinement lemmas in \calf{} displayed in inference rule style. }
\label{fig:refinement-lemmas}
\end{figure}

\subsection{Recursion}\label{sec:recursion}

As mentioned in \cref{sec:general-recursion}, \citeauthor{bove-capretta:2005}'s accessibility predicates
provide a way to express general recursive definitions in type theory.
Inspired by \citet{niu-harper:2020},
we provide an alternative approach in \calf{} that exploits the cost structure of programs:
one can use the cost bound of a given algorithm to \emph{safely define} the algorithm in question.
Instead of accessibility predicates, we may parameterize every program by a \emph{clock} that
represents the amount of fuel available for recursion. We say that an instantiation of the clock
is \emph{safe} when it provides enough fuel for the clocked program to satisfy the behavioral
specification of the algorithm. By definition the \emph{recursion depth} of the program is
a safe instantiation.

Furthermore, note that the \emph{cost} of the program is an upper bound on the recursion depth in many cost models.
In such cases defining an algorithm in \calf{} is intertwined with extracting and verifying its cost bound,
evincing a synergy one enjoys in the cost-aware setting:
algorithms with interesting cost structure require general recursive definitions,
meanwhile their safety as clocked programs is derived from the cost bound.
Observe that this paradigm is a legitimate encoding of general recursion
because \emph{we do not track the cost of computing the cost bound}.
One may think of this arrangement as programming with a version of for loops
whose bounds are computed in a cost-free manner.

\paragraph{Method of recurrence relations}

To put the plan in action, we outline a recipe for defining and analyzing an algorithm using the method of recurrence relations in \calf{}:

\begin{enumerate}
  \item An algorithm is given along with its \emph{cost model}. Place $\mathsf{step}$
  in accordance with the cost model to obtain a cost-aware instrumentation of the algorithm.
  \item Define a \emph{clocked} version of the algorithm; explicitly, one parametrizes the algorithm by an extra clock argument of type $\nat$
  representing the available fuel; when the clock is nonzero, the program follows the designated recursion pattern by
  decrementing the clock, and when the clock is zero, the program terminates by returning a default value or raising an exception.
  \item Define the \emph{recursion depth} that bounds the number of recursive calls.
  Because we do not track the cost of computing the recursion depth, it may be defined however convenient.
  \item Define the the associated \emph{cost recurrence} that maps inputs and to costs.
  Often times this may be used in place of the recursion depth as it is an upper bound.
  Similar to the recursion depth, we do not track the cost of the cost recurrence.
  \item Obtain the \emph{complete program} by instantiating the clocked program with the recursion depth.
  Prove this is a safe instantiation in the sense that the
  resulting program satisfies the behavioral specification of the algorithm (\eg computes the greatest common divisor).
  \item Prove that the resulting algorithm is bounded by the cost recurrence. This process is mostly mechanical:
  one repeatedly applies the lemmas in \cref{sec:quantitative-refinements} to break down $\mathsf{isBounded}$ goals.
  \item Refine the recurrence by (\eg) computing a closed-form solution. Usually this step represents the bulk
  of the work in pen-and-paper algorithm analysis.
\end{enumerate}

We apply this recipe in the following section to analyze Euclid's algorithm for the
greatest common divisor.

\section{Verification in \calf{}}\label{sec:verification}

We demonstrate in \calf{} two fundamental techniques used pervasively in algorithm analysis.
First, we illustrate the \Alert{method of recurrence relations}
by analyzing Euclid's algorithm for the greatest common divisor, proving
its correctness and deriving an asymptotically tight upper bound on the number of modulus operations used.
Second, we formalize physicist's method for \Alert{amortized analysis}
by studying the complexity of sequences of \emph{batched queue} operations, verifying
that each queue operation has constant amortized cost.
Due to space limitations we cannot discuss the sorting case study in detail, but
we mention results concerning parallel complexity in \cref{sec:parallelism},
and the interested reader can find the full development in the supplementary material.

Through these case studies, we promote a comprehensive verification pipeline made possible
by the unification of the following ingredients in a single framework:
\begin{enumerate}
  \item Specification of cost models
  \item Formal connection between algorithms and their associated recurrence relations
  \item A modality that administers extensional properties
  \item Full-spectrum dependent types that provides a rich specification language
\end{enumerate}

\paragraph{Cost models}

Prior to analyzing an algorithm, one has to make clear what ``counts'' as cost.
A particularly simple definition is to count every transition step in an
operational semantics, resulting in a \emph{language-level} cost semantics.  On
the other hand, algorithms researchers prefer a different perspective in which
cost is an \emph{algorithm-specific} notion.
For example, a common cost model for sorting algorithms counts the number of
comparisons, which does not account for the cost of (\eg) constructing lists.
This view allows one to study the underlying combinatorial structure of an
algorithm without getting distracted by implementation details. This is the
prevailing perspective we take in \calf{}, although one can also work with a
uniform language-level cost semantics when necessary; for instance, in the
amortized analysis of batched queues (see \cref{sec:queues}) we axiomatize a type of \emph{cost-aware}
lists that charges one $\mathsf{step}$ per recursive call.

\paragraph{Formalizing recurrence extraction}

Recurrence relations are a fundamental concept in algorithm analysis --- every
algorithm can be \emph{abstracted} into an associated cost recurrence that
characterizes the relationship between the input and the induced cost.  Recent
work of \citet{kavvos-morehouse-licata-danner:2019} has provided mathematical
grounding for informal proofs involving recurrence relations in the form of a
verified procedure for extracting (higher-order) recurrence relations from CBPV
programs.  Although \calf{} does not support recurrence extraction in the
mechanical style proposed by \citet{kavvos-morehouse-licata-danner:2019}, one
can manually define a recurrence and express its relationship to the given
algorithm by proving the \emph{internal} $\mathsf{isBounded}$ refinement.
Indeed, one of the advances embodied in \calf{} is the unification
of the distinct phases/languages in \citet{kavvos-morehouse-licata-danner:2019}
into a single framework that furnishes a programming language with support for cost
specification.

\paragraph{Managing extensionality}
As discussed in \cref{sec:lang-of-phase-distinctions},
the language of phase distinctions naturally induces a modality $\Op$ for extension,
which we use to express behavioral specifications in \calf{}. For instance,
we express the correctness of Euclid's algorithm by proving that
it satisfies the characteristic equations of the gcd under the extensional modality $\Op$.

\paragraph{Cost-aware logical framework}

Decades of experience has shown the effectiveness of using dependent type
theories to encode mathematics \citep{buzzard-commelin-massot:2020,
han-van-doorn:2020, gonthier:2008} and to verify \emph{behavioral}
properties of programs \citep{chlipala:2013, lee-crary-harper:2007,
ullrich:2016, stump:2016}.  Our experience with \calf{} suggests
that dependent type theories are \emph{also} an appropriate tool for analyzing
intensional properties of programs including cost.  In the following case
studies we rely on the rich type structure of \calf{} to evaluate different
strategies for establishing cost bounds.  We stress that \calf{} is a
\emph{framework} for quantitative reasoning: instead of working with a fixed
set of rules, one is free to choose the most appropriate tool for the given
problem.

\subsection{Euclid's algorithm}\label{sec:gcd}
In our first case study we analyze Euclid's algorithm for calculating the greatest common divisor, the
prototypical example of an algorithm that relies on nonstructural recursion. Our analysis closely follows the steps in the recipe from \cref{sec:recursion}.

\paragraph{Behavioral specification}
Let \isof{\textit{gcd}}{\tmv{\nat^2} \to \tmc{\nat}} be a candidate \calf{} program for
computing the gcd. Inspired by the usual formulation of Euclid's algorithm,
we may specify the correct \emph{behavior} of \textit{gcd} with the following propositions:
\begin{align*}
  \Op(\textit{gcd}\,(x,\zero) &= \ret{x}) \tag{3}\label{eq:gcd-1}\\
  \Op(\textit{gcd}\,(x,\suc{y}) &= \textit{gcd}\,(\suc{y}, \textit{mod}\,(x, \suc{y}))) \tag{4}\label{eq:gcd-2}
\end{align*}

Above we have assumed that there is a (cost-free) \calf{} program
\isof{\textit{mod}}{\tmv{\nat^2} \to \tmv{\nat}} that computes the modulus.  In
other words \cref{eq:gcd-1,eq:gcd-2} state that \textit{gcd} satisfies the
defining clauses of Euclid's algorithm in the extensional fragment.

\paragraph{Specializing the cost structure}
Because the gcd is defined on the natural numbers, we instantiate the cost structure $\mathbb{C}$ at
the ordered monoid $(\Nat, +, 0, \le)$.

\begin{small}
\begin{figure}
  \centering
  \begin{minipage}[t]{\textwidth}
    \begin{align*}
      \textit{mod}_{\mathsf{inst}} &: \tmv{\nat} \to \tmv{\nat} \to \tmc{\F{\nat}}\\
      \modnum{x}{y} &= \Alert{\mathsf{step}^1}(\ret{\textit{mod}\,(x, y)})
    \end{align*}
  \end{minipage}
 \begin{minipage}[t]{\textwidth}
  \begin{align*}
    \gcdclocked &: \tmv{\nat} \to \tmv{\nat^2} \to \tmc{\F{\nat}}\\
    \gcdclocked\, (\zero) &= \lambda (x,y).\, \ret{x}\\
    \gcdclocked\, (\suc{k}) &= \lambda (x,y).\, \reccst{y}{\ret{x}}{\lambda y',\_.\, e_1}\\
    e_1 &= \bindcst{r}{\modnum{x}{\suc{y'}}}{\gcdclocked(k)(\suc{y'}, r)}
  \end{align*}
\end{minipage}
\begin{minipage}{0.5\textwidth}
  \begin{align*}
    \gcdcostnat &: \tmv{\nat^2} \to \tmv{\nat} \\
    \gcdcostnat (x, y) &=
    \begin{dcases*}
      \zero &\text{if} y = \zero\\
      \suc{\gcdcostnat(y, \textit{mod}\,(x, y))} &\text{ o.w.}
    \end{dcases*}
  \end{align*}
\end{minipage}
\begin{minipage}{0.5\textwidth}
  \begin{align*}
  \gcdcode &: \tmv{\nat^2} \to \tmc{\F{\nat}}\\
  \gcdcode\, (x,y) &= \gcdclocked\,(\gcdcostnat\,(x,y))(x,y)\\
\end{align*}
\end{minipage}

\caption{Euclid's algorithm in \calf{}. From top to bottom: $\textit{mod}_{\mathsf{inst}}$ is the cost instrumented modulus operation,
\gcdclocked{} is the clocked algorithm, \gcdcostnat{} is the recursion depth/cost recurrence, and \gcdcode{} is the
final program. Note that because \gcdcostnat is cost-free,
we may define it however convenient, \eg by well-founded induction on the arguments.
}
\label{fig:gcd-code}
\end{figure}
\end{small}

\paragraph{Executing the recipe}

We execute the recipe from \cref{sec:recursion} to analyze Euclid's algorithm.
The associated \calf{} programs are displayed in \cref{fig:gcd-code}.
First we define the \Alert{cost model} to be the number of \textit{mod} operations, encoded in the instrumented version of the modulus, $\textit{mod}_{\mathsf{inst}}$,
which is used to define the \Alert{clocked} gcd algorithm $\gcdclocked$.
Here the first parameter of $\gcdclocked$ serves as the termination metric:
recursive calls in Euclid's algorithm are justified by decrementing the clock parameter.
Next, observe that under our cost model the \Alert{recursion depth} and \Alert{cost recurrence}
coincide for Euclid's algorithm, so we define a single (cost-free) program
$\gcdcostnat$ that simultaneously provides a sufficient instantiation (described in \cref{sec:recursion})
of the clock in $\gcdclocked$ and a cost recurrence for the algorithm.
Consequently the \Alert{complete algorithm} $\gcdcode$ is obtained by instantiating the clock parameter in
$\gcdclocked$ with $\gcdcostnat$. We prove that $\gcdcode$ correctly implements gcd:

\begin{theorem}[\Mech{Examples.Gcd.Spec.\{gcd≡spec/zero, gcd≡spec/suc\}}]\label{theorem:gcd-safety}
  We have that $\gcdcode$ behaves correctly, \ie \cref{eq:gcd-1,eq:gcd-2} hold
  for all \isof{x,y}{\tmv{\nat}}.
\end{theorem}

Let $\iota$ be the obvious isomorphism $\tmv{\nat} \cong \Nat$ \footnote{Because
both the cost monoid $\Nat$ and the data type $\nat$ is defined via the Agda
type \AgdaPrimitiveType{$\Nat$}, $\iota$ is the identity in our
implementation.}.  We verify that $\gcdcode$ is bounded by $\iota
\circ \gcdcostnat$:

\begin{theorem}[\Mech{Examples.Gcd.Clocked.gcd≤gcd/cost}]~\label{theorem:gcd/clocked-gcd/cost}
  For all \isof{x,y}{\tmv{\nat}}, we have that \isBounded{\nat}{\gcdcode\,(x,y)}{(\iota \circ \gcdcostnat)(x,y)},
\end{theorem}

Lastly we prove a refinement for the recurrence \gcdcostnat{} by computing a closed-form bound.
Let $\mathsf{Fib} : \Nat \to \Nat$ be the fibonacci sequence, and let $\mathsf{Fib}^{-1} : \Nat \to \Nat$
be the function characterized by the equation
$\mathsf{Fib}^{-1}(x) = \max{\{i \mid \mathsf{Fib}(i) \le x \}}$.
Note that $\mathsf{Fib}^{-1}$ is well-defined since $\mathsf{Fib}$ is strictly monotonic for $n \ge 2$.
It is well-known that the cost bound $\iota \circ \gcdcostnat{}$ is closely related to $\mathsf{Fib}^{-1}$:

\begin{theorem}[\Mech{Examples.Gcd.Refine.gcd/cost≤gcd/cost/closed}]\label{theorem:gcd/cost-closed}
  For all \isof{x,y}{\tmv{\nat}}, we have that $(\iota \circ \gcdcostnat)(x,y) \le \mathsf{Fib}^{-1}(\iota(x)) + 1$.
\end{theorem}

\begin{corollary}[\Mech{Examples.Gcd.Refine.gcd≤gcd/cost/closed}]
  For all \isof{x,y}{\tmv{\nat}}, we have that \isBounded{\nat}{\gcdcode\,(x,y)}{\mathsf{Fib}^{-1}(\iota(x))+1}.
\end{corollary}

\subsection{Amortized analysis}\label{sec:queues}

In addition to the method of recurrence relations, we may formulate more advanced algorithm analysis techniques.
As an example, we illustrate the \calf{} formalization of \emph{amortized analysis}.
First introduced by \citeauthor{tarjan:1985} in the mid-80s, amortized analysis is a
method to establish cost bounds on \emph{sequences} of operations on a data structure that is more
precise than a simple union bound. In this section we present a version of amortized analysis
known as the physicist's method: given a data structure $s$, one may define a measure $\Phi : s \to \mathbb{Z}_+$
that represents the amount of \emph{potential} that can be used to do work. The crux of the analysis
is to rig $\Phi$ so that expensive operations are associated with large decreases
in potential; because $\Phi$ is nonnegative, this ensures that expensive operations cannot occur
too often in a given sequence, \ie their cost is \emph{amortized}.

\paragraph{Batched queues}
To illustrate the physicist's method, we analyze the amortized complexity of a queue implementation known as \emph{batched queues}
\citep{gries:1987, hood-melville:1981, burton:1982, okasaki:1998}. A batched queue is
a pair of lists $(f,b)$ coupled with the invariant that the logical order of the queue is $f :: \textit{rev}(b)$.
The \calf{} implementation of the batched queue is presented in \cref{fig:batched-queues}
\footnote{Axiomatization of the unit type and sum type can be found in \cref{fig:calf-types} in \cref{sec:calf-def}.}.
For simplicity, we only consider elements of type $\nat$.

\paragraph{Specializing the cost structure}
For amortized analysis of batched queues, we instantiate the cost monoid $\mathbb{C}$ at the ordered monoid
$(\mathbb{N}, +, 0, \le)$ whose structure as a semiring and compatibility with the integers $\mathbb{Z}$
are required to define and reason about the potential function.

\begin{figure}
  \begin{minipage}[t]{0.3\textwidth}
    \begin{align*}
      Q &\coloneqq \calist{1}{\nat} \times \calist{1}{\nat}
    \end{align*}
  \end{minipage}
  \begin{minipage}[t]{0.3\textwidth}
    \begin{align*}
      \textit{enq} &: \tmv{Q} \to \tmv{\nat} \to \tmv{Q}\\
      \enq{(f,b)}{x} &= (f, \consex{x}{b})
    \end{align*}
  \end{minipage}
  \begin{minipage}{0.3\textwidth}
    \begin{align*}
      \textit{deq}_0 &: \tmv{\listty{\nat}} \to \tmc{\F{\sumty{\unit}{Q \times \nat}}}\\
      \deqemp{b} &= \bindcst{l}{\rev{b}}{\listreccst{l}{\dots}{\ret{\inl{\triv}}}{\lambda a,l',\_.\, \ret{\inr{(l', \nilex), a}}}}
    \end{align*}
  \end{minipage}
  \begin{minipage}{0.3\textwidth}
    \begin{align*}
      \textit{deq} &: \tmv{Q} \to \tmc{\F{\unit + (Q \times \nat)}}\\
      \deq{(f,b)} &= \listreccst{f}{\lambda \_.\, \sumty{\unit}{Q \times \nat}}{\deqemp{b}}{\lambda a,f',\_.\, \ret{\inr{(f',b), a}}}
    \end{align*}
  \end{minipage}
  \caption{Batched queues in \calf{}.}
  \label{fig:batched-queues}
\end{figure}

\paragraph{Cost model}
A common cost model in this setting is the number of list iterations. We encode this cost model by axiomatizing
a type of cost-aware lists $\mathsf{L} : \Alert{\mathbb{C}} \to \tpv \to \tpv$
\footnote{Full definition can be found in \cref{fig:calf-types-1} in \cref{sec:calf-def}. }
that is parameterized by the amount to charge for each recursive call. The type $\mathsf{L}$ has
the standard constructors $\mathsf{nil}$ and $\mathsf{cons}$; the only new rule is the
destruction of cons nodes, which induces the annotated amount of cost:
\begin{align*}
  \mathsf{rec}/\mathsf{cons} &: \impl{c,A,a,X,e_0,e_1} (\isof{l}{\tmv{\calist{c}{A}}}) \to\\
  &\listrec{\consex{a}{l}}{X}{e_0}{e_1} = \Alert{\mathsf{step}^c}(e_1(a)(l)(\listrec{l}{X}{e_0}{e_1}))
\end{align*}
To charge unit cost per iteration, we define the type of batched queues as
$Q \coloneqq \calist{1}{\nat} \times \calist{1}{\nat}$.
Note that the standard list type is recovered as $\listty{A} \coloneqq \calist{0}{A}$.
We write $\len{\--} : \impl{c} \calist{c}{A} \to \Nat$ for the length function on lists.

\paragraph{Upper bounding individual queue operations}
We obtain cost bounds on the individual operations using similar techniques as in \cref{sec:gcd}:

\begin{theorem}[\Mech{Examples.Queue.enq≤enq/cost}, \Mech{Examples.Queue.deq≤deq/cost}]\label{thm:cost-ops}
  For any \\ queue $q$ and element $x$, we have \isBounded{Q}{\enq{q}{x}}{0}.
  Moreover, for any queue $q = (f,b)$, we have \isBounded{1 + Q \times \nat}{\deq{q}}{1 + \len{b}}.
\end{theorem}

\paragraph{Serializing the queue operations}
To formalize the notion of a sequence of operations, we define a serialization of the queue operations in \cref{fig:op-serialize}.
Here, $\op$ denotes the type of queue operations, which is either an enqueue of an element or a dequeue.
Given a serialized operation $o$ and a queue $q$, \operate{o}{q} is the interpretation of $o$ on $q$.
By \cref{thm:cost-ops} the resulting computation is bounded by the
cost of the corresponding operation \cost{q}{o}, defined in \cref{fig:op-serialize}:
\begin{corollary}[\Mech{Examples.Queue.op≤op/cost}]
  Given an operation $o$ and a queue $q$, we have \isBounded{Q}{\operate{o}{q}}{\cost{q}{o}}.
\end{corollary}
The function $\operateseq{\--}{\--}$ lifts the interpretation to sequences of operations by threading the
given queue through the list of operations. It is bounded by $\textit{cost}_{\mathsf{seq}}$:
\begin{lemma}[\Mech{Examples.Queue.op/seq≤op/seq/cost}]\label{lemma:seq-cost}
  Given a list of operations $l$ and a queue $q$, we have \isBounded{Q}{\operateseq{l}{q}}{\costseq{l}{q}}.
\end{lemma}

\begin{small}

\begin{figure}
  \begin{minipage}[t]{0.25\textwidth}
    \begin{align*}
  \op &: \tpv\\
  \op &= \sumty{\nat}{\unit}\\
  \openq{x} &= \inl{x}\\
  \opdeq &= \inr{\triv}
    \end{align*}
  \end{minipage}
  \begin{minipage}[t]{0.25\textwidth}
    \begin{align*}
    \textit{cost} &: \tmv{\op} \to \tmv{Q} \to \mathbb{Z}\\
    \cost{\openq{x}}{q} &= 0\\
    \cost{\opdeq}{(f,b)} &= 1 + \len{b}
  \end{align*}
  \end{minipage}
  \begin{minipage}{0.3\textwidth}
   \begin{align*}
  \operate{\--}{\--} &: \tmc{\op \to Q \to \F{Q}}\\
  \operate{\openq{x}}{q} &= \enq{q}{x}\\
  \operate{\opdeq}{q} &= s \leftarrow \deq{q}; \\
  \mathsf{case}(s)\{\; &\inl{\triv} \hookrightarrow \ret{(\nilex,\nilex)}\\
  \mid\; &\inr{(q,x)} \hookrightarrow \ret{q} \;\}
  \end{align*}
  \end{minipage}
  \begin{minipage}{0.3\textwidth}
    \begin{align*}
  \operateseq{\--}{\--} &: \tmc{\listty{\op} \to Q \to \F{Q}}\\
  \operateseq{\nilex}{q} &= \ret{q}\\
  \operateseq{\consex{o}{\textit{os}}}{q} &= \bindcst{q'}{\operate{o}{q}}{f(q')}
\end{align*}
  \end{minipage}
  \begin{minipage}{0.3\textwidth}
    \begin{align*}
    \textit{cost}_{\mathsf{seq}} &: \tmv{\listty{\op}} \to \tmv{Q} \to \mathbb{Z}\\
    \costseq{\nilex}{q} &= 0\\
    \costseq{\consex{o}{\textit{os}}}{q} &= \cost{o}{q} + (\bindcst{q'}{\operate{o}{q}}{\costseq{os}{q'}})
  \end{align*}
  \end{minipage}
  \caption{Serialization of queue operations.}
  \label{fig:op-serialize}
\end{figure}
\end{small}

\paragraph{Amortized analysis of batched queues}
We are now in a position to analyze the amortized cost of the queue operations.
We define the potential function on queue states:
\begin{align*}
  \Phi &: \tmv{Q} \to \mathbb{N}\\
  \Phi(f,b) &= \len{f} + 2\cdot \len{b}
\end{align*}

Traditionally an operation's amortized cost is defined as the maximum value of the sum of the induced cost
and the difference in the potential over a starting state; we represent this relationally:
\begin{align*}
  \mathsf{hasCost}_{\mathsf{amortized}} &: \tmv{\op} \to \Nat \to \jdg\\
  \acost{o}{k} &= (\isof{q}{Q}) \to \left(\cost{o}{q} +_{\mathbb{Z}} \Phi(\operate{o}{q}) -_{\mathbb{Z}} \Phi(q)\right) \le_{\mathbb{Z}} k
\end{align*}

Note that because amortized cost has to be defined using \emph{non-truncated} subtraction,
terms of type $\Nat$ appearing in the relation $\mathsf{hasCost}_{\mathsf{amortized}}$
are all implicitly lifted to the integers $\mathbb{Z}$.
We verify that the amortized cost of enqueue is 2, while the amortized cost of dequeue is 0:

\begin{theorem}[\Mech{Examples.Queue.enq/acost}, \Mech{Examples.Queue.deq/acost}]
  We have that\\ \acost{\openq{x}}{2} for all \isof{x}{\tmv{\nat}} and that \acost{\opdeq}{0}.
\end{theorem}

Using the amortized costs, we can bound the cost of a sequence of queue operations using a standard telescoping series:
\begin{theorem}[\Mech{Examples.Queue.op/seq/cost≤ϕ₀+2*|l|}]
  Given an initial queue \isof{q}{\tmv{Q}} and a list of operations \isof{l}{\tmv{\listty{\op}}},
  we have $\costseq{l}{q} \le \Phi(q) + 2 \len{l}$.
\end{theorem}
Combining this inequality with \cref{lemma:seq-cost}, we obtain an amortized bound on a sequence of operations on the empty queue:
\begin{corollary}[\Mech{Examples.Queue.op/seq≤2*|l|}]
  Given a list of operations $l$, we have that\\
  \isBounded{Q}{\operateseq{l}{(\nilex,\nilex)}}{2\len{l}}.
\end{corollary}

\section{Metatheory of \calf{}}\label{sec:metatheory}

In this section we substantiate the theory of \calf{} by means of a model construction
and prove the following theorems:
\begin{enumerate}
  \item \emph{\textbf{Nondegeneracy.}}
  The cost effect $\mathsf{step}$ is not degenerate, i.e $\nvdash \mstep{1}{e} = \ret{e}$ for any \isof{e}{\F{A}}.
  \item \emph{\textbf{Validity of cost bounds}}.
  We have that $\vdash \Op{(m \le n)}$ if and only if $\vdash m \le n$ for all \isof{m,n}{\mathbb{\Nat}}.
\end{enumerate}

\paragraph{Models of \calf{}}
Recall from \cref{sec:calf} that we define \calf{} as the \emph{free} lccc $\CCat_{\calf{}}$
over the signature $\Sigma_{\calf}$ presented in \cref{fig:calf}.
Consequently one may prove metatheorems about \calf{} using the universal property of freely generated categories.
In the context of functorial semantics \citep{lawvere:thesis},
the universal property states that one may define a model $\CCat_{\calf} \to \ECat$ by simply specifying the image of the
constants of $\Sigma_{\calf{}}$ in $\ECat$ \footnote{An analogous situation arises when considering homomorphisms out of a free group: \emph{any} function on the generators determines a homomorphism.}.
The data of this specification is encapsulated by the notion of an \emph{algebra} for a signature:

\begin{definition}[Algebra for a signature in the logical framework]\label{def:algebra}
  Let $\ECat$ be a category that has a universe $\mathcal{U}$ closed under dependent products, dependent sums, and extensional equality.
  Given a signature $\Sigma$ in the logical framework, we can define a type $\textbf{Alg}_{\mathcal{U}}(\Sigma)$ of $\mathcal{U}$-small algebras for $\Sigma$ in
  $\ECat$ by interpreting \jdg{} as $\mathcal{U}$ and taking the
  dependent sum over all the constants declared in $\Sigma$.
\end{definition}

Thus given a sufficiently structured category $\ECat$ in the sense above,
we can define a model of \calf{} by exhibiting an algebra \isof{\mathcal{A}}{\textbf{Alg}_{\mathcal{U}}(\Sigma_{\calf})}
in some universe $\mathcal{U}$ of $\ECat$. In fact we can define the intended model of \calf{} in \emph{any} given
topos $\TopIdent{X}$ with a distinguished subterminal object representing the phase separation of intension and extension.
To obtain an external view, we specialize the construction to the presheaf topos over the interval category $\{0 \to 1\}$, \ie
the category of families of sets $\textbf{Set}^\to$, which suggests the interpretation of \calf{} types as phase separated \emph{families}.

\paragraph{Language of phase distinctions}
Inspired by recent work emphasizing the
role of phase distinctions in the analysis of metatheoretic properties~\citep{sterling-harper:2021,sterling-angiuli:2021},
we isolate a pair of complementary modalities $\Op, \Cl$ that models the phase distinction of
extension and intension in \calf{}. Using the language of phase distinctions,
we give a succinct definition of our model that avoids the explicit but more cumbersome presentation involving families.

\subsection{Counting model of \calf{}}

We exhibit an algebra $\mathcal{A}$ for $\Sigma_{\calf}$ in any given topos $\TopIdent{X}$ equipped with a distinguished proposition $\isof{\ExtOpn}{\Omega}$.\footnote{For the limited topos theory we require in this section, we employ the notations of \citet{anel-joyal:2021}.}
Consequently we have at our disposal a rich internal language in the form of
an extensional dependent type theory that includes (in particular)
a hierarchy of universes $\mathcal{U}_{\alpha}$, inductive types, and a universe of proof-irrelevant propositions $\Omega$.
The role of the proposition $\ExtOpn$ is to provide a semantic counterpart to the \calf{} proposition $\ExtOpn$.

Letting $\alpha < \beta$ be universe levels, we then define an algebra \isof{\mathcal{A}}{\textbf{Alg}_{\mathcal{U}_\beta}(\Sigma_{\calf})}
that constitutes the standard Eilenberg–Moore model of CBPV in which computation types are interpreted as
\emph{algebras}~\footnote{Not to be confused with \cref{def:algebra}.} for a given monad. In the case of \calf{}
we dub this interpretation the \emph{counting model}, so named because the interpretation of the computation type $\F{A}$
is the free algebra of a particular writer monad whose carrier classifies elements of $A$ paired
with a step count.
Because many parts of the interpretation is standard, we highlight only the constructions pertaining to \calf{} per se.

\subsubsection{Phase distinction}
As mentioned above, we define the extensional phase $\ExtOpn$ as the distinguished proposition $\ExtOpn$.
By definition, the extensional modality is rendered as the function space in the internal language of $\TopIdent{X}$, i.e
$\Op{\--} \coloneqq \ExtOpn \to \--$.
The intensional modality $\Cl{\--}$ is defined as the pushout $A \sqcup_{A \times \ExtOpn} \ExtOpn$
of the projection maps of $A \times \ExtOpn$.

\begin{proposition}[\citet{rijke-shulman-spitters:2017}]
  Both $\Op,\Cl$ are idempotent, left exact and monadic.
\end{proposition}

We write $(\eta_{\Op}, \eta_{\Cl})$ for the monadic unit of the (extensional, intensional) modality.
Observe that $\Cl{A}$ collapses to a single point when a proof of $\ExtOpn$ exists:

\begin{proposition}\label{lemma:closedwhenopen}
  Given \isof{u}{\ExtOpn}, we have that $\closed{A} \cong 1$ for any $A$.
\end{proposition}

Thus we may effect the erasure of $\mathsf{step}$ in the extensional fragment
by arranging the cost structure of programs to be a type in the \emph{image} of $\Cl$:
when a proof \isof{u}{\ExtOpn} is present, a cost $c$ is equal to any other cost,
in particular $0$; consequently we have $\mstep{c}{e} = \mstep{0}{e} = e$ by the
coherence of $\mathsf{step}$.

\subsubsection{Cost monoid $\mathbb{C}$}

Recalling that \calf{} is parameterized in a cost monoid $\mathbb{C}$,
our model takes as an input an arbitrary $(\mathbb{M},+,0,\le)$ cost monoid in the category of sets $\SET$.
We then define $\mathbb{C}$ as the image of $\mathbb{M}$ under the constant sheaf functor $\SET \to \Sh{\TopIdent{X}}$.
Note that because $\mathbb{C}$ is not
necessarily in the image of $\Cl$, we interpret computation types of \calf{} as
algebras for the writer monad $\Cl\mathbb{C} \times \--$.
By \cref{lemma:closedwhenopen} the cost structure of programs is then rendered trivial underneath $\ExtOpn$.

\subsubsection{Judgmental structure}

Per the Eilenberg–Moore model of CBPV, value types \calf{} are simply interpreted as types in $\TopIdent{X}$,
and computation types are interpreted as algebras for $\Cl\mathbb{C} \times \--$:

\begin{minipage}[c]{0.4\textwidth}
  \begin{align*}
  &\alg{T} =\\
  &\begin{cases}
    \mathsf{A} : \mathcal{U}_{\alpha}\\
    \mathsf{map} : T (\mathsf{A}) \to \mathsf{A}\\
    \mathsf{unit} : \mathsf{map} \circ \eta = id_{\mathsf{A}}\\
    \mathsf{mult} : \mathsf{map} \circ \mu = \mathsf{map} \circ T \mathsf{map}
  \end{cases}
\end{align*}
\end{minipage}
\begin{minipage}[c]{0.25\textwidth}
  \begin{align*}
  \tpv &: \mathcal{U}_{\beta}\\
  \tpv &= \mathcal{U}_{\alpha}\\
  \tmv{A} &= A
  \end{align*}
\end{minipage}
\begin{minipage}[c]{0.25\textwidth}
  \begin{align*}
    \tpc &: \mathcal{U}_{\beta}\\
    \tpc &=\alg{\Cl{\mathbb{C}} \times \--}\\
    \tmc{X} &= \carrier{X}
  \end{align*}
\end{minipage}\\

Note that given an algebra $X$, we write \carrier{X} for the carrier
$X \cdot \mathsf{A}$.

\subsubsection{Values and computations}
In the algebra semantics of CBPV, one coerces between value types and computation types via the
adjoint pair $\mathsf{F} \dashv \mathsf{U}$ in which the left adjoint
takes a type to the associated free $\Cl{\mathbb{C}} \times \--$-algebra and the right
adjoint forgets the structure of the given algebra, writing $\freealg{T}{A}$ for the free $T$ algebra on $A$:

\begin{minipage}{0.45\textwidth}
  \begin{align*}
  \mathsf{F} &: \mathcal{U}_{\alpha} \to \alg{\mathbb{C} \times \--}\\
  \F{A} &= \freealg{\Cl{\mathbb{C}} \times \--}{A}
  \end{align*}
\end{minipage}
\begin{minipage}{0.45\textwidth}
  \begin{align*}
  \mathsf{U} &: \alg{\Cl{\mathbb{C}} \times \--} \to  \mathcal{U}_{\alpha}\\
  \UU{X} &= \carrier{X}
\end{align*}
\end{minipage}

\subsubsection{Cost effect}
The cost effect $\mathsf{step}$ is given by the algebra map of the given computation type:
\begin{align*}
  \mathsf{step} &: \impl{X} \mathbb{C} \to \carrier{X} \to \carrier{X}\\
  \cstep{c}{x} &= (X \cdot \mathsf{map}) (\eta_{\Cl}(c), x)
\end{align*}

The following is an immediate consequence of \cref{lemma:closedwhenopen}.
\begin{corollary}[Extensional fragment]
  We have that $\Op{(\mstep{c}{e} = e)}$ for all \isof{c}{\mathbb{C}} and \isof{e}{\carrier{X}}.
\end{corollary}

\paragraph{External view of the counting model}
We may obtain a more concrete perspective on the counting model
by considering its construction in the arrow category $\textbf{Set}^{\to}$
in which the extensional phase $\ExtOpn$ is furnished by the subterminal family $0 \to 1$.
Observe that objects in this category are
families of sets $A : A_1 \to A_0$, which corresponds to the fact that a type $A$ is a family indexed in a collection
of \emph{behaviors} with the fibers representing the \emph{cost structure} for a given behavior.

In $\textbf{Set}^{\to}$ the extensional modality takes a family $A_1 \to A_0$ to the identity $A_0 \to A_0$,
trivializing the fiber (\ie cost structure) over each point in $A_0$.
On the other hand, the intensional modality takes $A_1 \to A_0$ to the family $A_1 \to 1$;
applying the extensional modality thence results in the terminal family $1 \to 1$,
demonstrating the property that ``the extensional part of the intensional part is trivial''.

\subsection{Nondegeneracy of $\mathsf{step}$}

\begin{theorem}\label{thm:non-degen}
  We have that $(\mstep{c}{e} = e) \to \closed{\bot}$ for any nonzero \isof{c}{\mathbb{C}} and
  \isof{e}{\Cl{\mathbb{C}} \times A}.
\end{theorem}

\begin{proof}
  By definition, $e = (c', a)$ for some \isof{c'}{\Cl{\mathbb{C}}} and \isof{a}{A}.
  Unfolding the definition of $\mathsf{step}$ and free algebra, we have $\mstep{c}{c',a} = (\eta_{\Cl}(c) +_{\Cl} c', a)$,
  where $+_{\Cl}$ lifts $+$ using the functorial action of $\Cl$.
  Hence it suffices to show $(\eta_{\Cl}(c) +_{\Cl} c', a) = (c', a)$ implies $\Cl{\bot}$.
  Suppose $(\eta_{\Cl}(c) +_{\Cl} c', a) = (c', a)$.
  By the induction principle of pushouts, there are two cases to consider.
  First, suppose $c' = \eta_{\Cl}(c'')$ for some
  \isof{c''}{\mathbb{C}}. Because $\Cl$ is left exact, the equation $\eta_{\Cl}(c) +_{\Cl} \eta_{\Cl}(c'') = \eta_{\Cl}(c'')$
  is equivalent to $\Cl(c+c'' = c'')$.
  But we assumed that $c$ is nonzero, so the fact that $\mathbb{C}$ is cancellative entails
  $c+c'' = c'' \to \bot$, and the result follows from the functorial action of $\Cl$.
  On the other hand, suppose $c' = \ast(u)$ for some \isof{u}{\ExtOpn}.
  By Lemma~\ref{lemma:closedwhenopen}, we obtain a unique proof of $\closed{\bot}$.
\end{proof}

Because $\closed{\bot} = \ExtOpn$, we know that if $\mathsf{step}$ is degenerate, then
the extensional phase $\ExtOpn$ is derivable. Observing that we placed no restrictions on the proposition
$\ExtOpn$ in the construction of the counting model, we immediately obtain the desired theorem
by instantiating $\ExtOpn$ with the false proposition:
\begin{theorem}
  We have that $\nvdash \cstep{c}{e} = e$ for any nonzero \isof{c}{\mathbb{C}} and
  \isof{e}{\F{A}}.
\end{theorem}

\subsection{Validity of extensional cost bounds}
We show that extensional inequalities are equivalent to ordinary inequalities in the $\SET^\to$ model of \calf{}
whenever the cost monoid is \emph{extensional} in the sense that $\mathbb{C} \cong \Op{\mathbb{C}}$
and the relation $\le$ may be characterized using $\Sigma$ and equality types.
As an example, we illustrate the case for the cost monoid $\Nat$,
noting that the same proof may be easily adapted to other common instances:

\begin{theorem}\label{thm:open-inequality}
  We have that $\Op{(m \le n)}$ if and only if $m \le n$ for all \isof{m,n}{\mathbb{\Nat}}.
\end{theorem}

\begin{proof}

  Observe that we may present $m \le n$ as the type $\Sigma \isof{k}{\Nat}.\, n = \mathsf{suc}^k\prn{m}$.
  By standard results \citep{rijke-shulman-spitters:2017}
  we know that the property of being extensional
  is closed under equality and $\Sigma$ types. Combined with the fact that $\Nat$ is an extensional type
  in the $\SET^\to$ model of \calf{}, we conclude that $m \le n$ is also extensional, \ie $(m \le n) \cong \Op{(m \le n)}$.
\end{proof}

\begin{corollary}
  We have that $\vdash \Op{(m \le n)}$ if and only if $\vdash m \le n$ for all \isof{m,n}{\mathbb{\Nat}}.
\end{corollary}

\section{Parallelism in \calf{}}\label{sec:parallelism}

Parallelism arises naturally in the setting of \calf{} via an equational presentation
of the profiling semantics of \citet{blelloch-greiner:1995}.
Here we present a version adapted from \citet{harper:2018:parallel} in which it is observed that
the source of parallelism can be isolated to the treatment of \emph{pairs} of computations:
a parallel computation of $A \times B$ is furnished by a new computation form \& that
conjoins two independent computations of $A$ and $B$:
\begin{align*}
  \& &: \impl{\isof{A,B}{\tpv}} \tmc{\F{A}} \to \tmc{\F{B}} \to \tmc{\F{A \times B}}
\end{align*}

One may think of a term \para{e}{f} as a computation in which
$e$ and $f$ are evaluated simultaneously.

\paragraph{Cost structure of parallelism}

\citet{blelloch-greiner:1995} characterize the complexity of a program in terms of two measures:
\emph{work}, which represents its sequential cost, and \emph{span}, which represents
its parallel cost.
In \calf{} this structure is recorded by the \emph{parallel cost monoid}
$\mathbb{C} \coloneqq (\Nat^2, \oplus, (0,0), \le_{\Nat^2})$
in which $\oplus$ and $\le_{\Nat^2}$ are component-wise extensions of addition and $\le$.
Parallel cost composition is then implemented by the operation
$(w_1,s_1) \otimes (w_2, s_2) \coloneqq (w_1 + w_2, \max{(s_1, s_2}))$
that takes the sum of the works and max of the spans.
This provides the required structure to assemble the cost of a completed parallel pair:
\begin{align*}
  \&_{\mathsf{join}} &: \impl{A,B,c_1,c_2,a,b}
    \para{\left(\mstep{c_1}{\ret{a}}\right)}{\left(\mstep{c_2}{\ret{b}}\right)} = \mstep{c_1 \otimes c_2}{\ret{(a,b)}}
\end{align*}

\paragraph{Nondegeneracy of parallel \calf{}}

Metatheoretic properties of parallel \calf{} follows directly from the counting model defined in \cref{sec:metatheory},
given that we can interpret parallel pairing. Because the new pairing operation is only defined on free algebras,
we may use $\otimes_{\Cl}$ (lift of $\otimes$ by the functorial action of $\Cl$) to define parallel pairing:
$\para{(c_1, a)}{(c_2, b)} = (c_1 \otimes_{\Cl} c_2, (a, b))$.

\paragraph{Parallel complexity of sorting}

We have verified the sequential and parallel complexity of insertion sort and merge sort under the comparison cost model.
As outlined above, we instantiate \calf{} with the parallel cost monoid $\Nat^2$ in which the first component
represents the sequential cost and the second component represents the parallel cost.
The analysis is parameterized by a \emph{comparable} type \isof{A}{\tpv} that is equipped with
a comparison operation \isof{\le^b}{\tmv{A} \to \tmv{A} \to \tmc{\F{\mathsf{bool}}}}.
Consequently, we may enforce the cost model by requiring the comparison operation $\le^b$ to be uniformly unit cost, \ie
$\isBounded{\mathsf{bool}}{x \le^b y}{(1,1)}$ for all \isof{x,y}{\tmv{A}}.
We have mechanized the following asymptotically tight cost bounds:

\begin{theorem}[\Mech{Examples.Sorting.Parallel.InsertionSort.sort≤sort/cost/closed}]
  For all \isof{l}{\listty{A}}, we have that \isBounded{\listty{A}}{\textit{isort}(l)}{(\len{l}^2,\len{l}^2)}.
\end{theorem}
Observe that sequential and parallel complexity coincide for insertion sort because there is
no opportunity for parallelism in the algorithm.
The standard merge sort algorithm enjoys a logarithmic speed up when the recursive calls are performed in parallel:
\begin{theorem}[\Mech{Examples.Sorting.Parallel.MergeSort.sort≤sort/cost/closed}]
  For all \isof{l}{\listty{A}}, we have that
  \isBounded{\listty{A}}{\textit{msort}(l)}{\left(\ceil*{\log_2{\len{l}}} \cdot \len{l}, 2 \cdot \len{l} + \ceil*{\log_2{\len{l}}}\right)}.
\end{theorem}
To obtain a sublinear bound on parallel complexity, one must modify merge sort to also perform the merging step in
parallel, an alteration that slightly increases the sequential complexity:
\begin{theorem}[\Mech{Examples.Sorting.Parallel.MergeSortPar.sort≤sort/cost/closed}]
  For all \isof{l}{\listty{A}}, we have that
  \isBounded{\listty{A}}{\textit{msortPar}(l)}{\left(\ceil*{\log_2{(\len{l} + 1)}}^2 \cdot \len{l}, \ceil*{\log_2{(\len{l} + 1)}}^3\right)}.
\end{theorem}

\section{Conclusion}

Three somewhat contradictory goals guide our type-theoretic approach to cost analysis:
\begin{enumerate}
  \item Expressiveness: the ability to codify the methods and results of informal algorithm analysis.
  \item Certification: programs and their cost bounds should bear their intended meaning.
  \item Composition: cost bounds should be composable.
\end{enumerate}

Most extant cost analysis frameworks excel at two out of three of the above.
Type systems defined by intrinsic cost-aware judgments \citep{hoffmann-aehlig-hofmann:2012, wang-wang-chlipala:2017, rajani-gaboardi-garg-hoffmann:2021}
are certified by soundness theorems and admit composition by construction but lack
expressiveness because cost bounds are often form-constrained and typing derivations
cannot exploit complex behavioral properties.
The traditional method for cost accounting using the writer monad \citep{handley-vazou-hutton:2019}
provides an expressive and compositional framework for cost analysis, but this transparent instrumentation
is not certified in the aforementioned sense because programs in the writer monad
do not necessarily accumulate cost faithfully (see \cref{sec:trans-abs}).
Lastly, frameworks for cost analysis in the setting of program logics \citep{atkey:2010, mevel-jourdan-pottier:2019} may be
transposed to type theory by working with a deep embedding of a programming language and its operational semantics
inside type theory. Although this can be developed into an expressive and certified framework in the sense above,
it is not compositional because one may speak about operational semantics only on closed terms
and must quantify over closing instances for open terms.

In this paper we show that the three goals may be achieved simultaneously.
First, the extensional fragment of \calf{} constituents an ordinary dependent type theory,
which furnished us a rich specification language to formulate two widely used algorithm
analysis techniques and illustrate each through detailed case studies.
Secondly, we see that \calf{} programs account for cost faithfully
because the type of free computations is \emph{abstract} in the theory,
leaving no opportunity for computations to spuriously abandon accumulated steps
or branch based on the cost component of an input.
Lastly, by axiomatizing cost as an effect in the base CBPV language, we obtain a
simple equational theory of cost that enables compositional cost analysis.
We conclude by suggesting two particularly pertinent directions for future investigations.

\paragraph{Automation}
In practice the usability of any verification framework may be greatly improved
by automating routine procedures or derivations. In the context of \calf{} there
are two immediate opportunities for automation. On the one hand,
the recurrence extraction step in the method of recurrence relations (as defined in \cref{sec:recursion})
may be automated in many cases by incorporating the mechanism of \citet{kavvos-morehouse-licata-danner:2019}.
On the other hand, proofs involving restricted forms of cost bounds may be automated either by recurrence solving (\eg
the Master theorem) or an automated system such as \raml{} \citep{hoffmann-aehlig-hofmann:2012}.

\paragraph{Full adequacy and partiality}
It would be interesting to prove an adequacy result of the form presented in \citet{kavvos-morehouse-licata-danner:2019}
in which one defines a cost-aware embedding of a source language (equipped with an operational semantics)
in the target language (in this case \calf{}) and proves that the image of any source program is assigned the same
cost as the cost of the source program that is induced by the operational semantics.
In many cases the source language of interest admits general recursion; consequently one must
arrange for \calf{} to faithfully interpret non terminating programs.
We conjecture that such an adequacy result may be proved by extending \calf{} with a version of
the partiality monad of \citeauthor{capretta:2005} and defining an embedding targeting the monadic fragment.

\nocite{niu-harper:2020}
\nocite{sterling-harper:2021}
\nocite{danner-licata-ramyaa:2015,kavvos-morehouse-licata-danner:2019}

\ifsubmission \else

\begin{acks}
  We are grateful to Carlo Angiuli and Alex Kavvos for productive discussions on the topic of
  this research, and to Tristan Nguyen at AFOSR for his support.

  This work was supported in part by AFOSR under grants MURI FA9550-15-1-0053 and
  FA9550-19-1-0216, in part by the National Science Foundation under award number CCF-1901381, and by AFRL through the NDSEG fellowship.
  Any opinions, findings and conclusions or recommendations
  expressed in this material are those of the authors and do not necessarily
  reflect the views of the AFOSR, NSF, or AFRL.
\end{acks}
\fi

\bibliography{references/refs-bibtex,bib}

\appendix
\section{Implementation in \agda{}}\label{sec:implementation}

We give a brief overview of the Agda implementation of \calf{};
the codebase and installation instructions are accessible through the supplementary materials.
We define \calf{} in Agda by postulating the constants in the signature $\Sigma_{\calf{}}$ (see \cref{fig:calf})
and animating the associated equations with the newly added rewriting facilities~\citep{cockx-tabareau-winterhalter:2021}.
For instance, the basic judgmental structure of \calf{} may be specified by the following \agda{} postulates and definitions:

\begin{minipage}[c]{0.25\textwidth}
\ExecuteMetaData[sections/code.tex]{lang-g1}
\end{minipage}
\begin{minipage}[c]{0.25\textwidth}
\ExecuteMetaData[sections/code.tex]{lang-g2}
\end{minipage}
\begin{minipage}[c]{0.25\textwidth}
\ExecuteMetaData[sections/code.tex]{lang-g3}
\end{minipage}
\begin{minipage}[c]{0.25\textwidth}
\ExecuteMetaData[sections/code.tex]{lang-g4}
\end{minipage}

An example showcasing the rewriting feature is the inversion principle for $\mathsf{bind}$:
\begin{center}
\ExecuteMetaData[sections/code.tex]{lang-g5}
\end{center}
Observe that the \agda{} implementation of \calf{} constitutes an algebra (see \cref{def:algebra})
$\mathcal{A}_{\text{Agda}}$ in which \AgdaPrimitiveType{Set} plays the role of the universe of judgments \jdg{}.

\paragraph{Importing Agda data types}
To maximize reuse of the Agda mathematical library, we axiomatize a new connective to import Agda data types into \calf{}:
\footnote{Note that we do not need to internalize the entire universe \AgdaPrimitiveType{Set}:
in actuality we only use \AgdaPrimitiveType{meta}
to import types pertaining directly to our case studies,
a situation that may be more objectively described by defining \AgdaPrimitiveType{meta} on a collection of
relevant \emph{type codes} that \emph{decode} into types in \AgdaPrimitiveType{Set}.
For the sake of simplicity, we elide this distinction in the implementation.}
\begin{center}
\ExecuteMetaData[sections/code.tex]{lang-meta}
\end{center}
For example, the type of \calf{} natural numbers maybe defined as
\AgdaPrimitiveType{nat} = \AgdaPrimitiveType{U (meta $\Nat$)}.
Note that we define \AgdaPrimitiveType{meta} as a connective that internalizes Adga data types as
\emph{computation} types, a design decision motivated by the fact that the connective is mostly
used in positions where a computation is expected, \eg when defining a function of the form
$\tmv{A} \to \tmc{\AgdaPrimitiveType{meta}(\Nat)}$.
Observe that \AgdaPrimitiveType{meta} generalizes the constants
$\widehat{\mathbb{C}}$ and $\widehat{\le}$ defined in \cref{fig:calf-cost}.

\paragraph{Uniform \emph{vs.} nonuniform cost models}
In \calf{}, cost analysis with respect to an \emph{algorithm-specific} cost model
may directly take advantage of the importing mechanism because it is trivial to instrument the
relevant operations with the desired cost \emph{after} they are defined in a cost-free manner in Agda.
For instance, in the analysis of Euclid's algorithm (see \cref{sec:gcd})
we may directly import the Agda type of natural numbers \AgdaPrimitiveType{$\Nat$}
and the associated modulus function, which we instrument in \calf{} as $\textit{mod}_{\mathsf{inst}}$.
In contrast, a \emph{uniform} cost model such as the one required in the amortized analysis of batched queues (see \cref{sec:queues})
precludes one from importing the required data type. In the case of batched queues,
it is not possible to define an Agda inductive type that charges $\mathsf{step}$'s for list iteration,
so we must axiomatize such a cost-aware list type in \calf{} (see \cref{fig:calf-types-1}).
\clearpage
\section{Complete definition of \calf{}}\label{sec:calf-def}

\begin{figure}[h]
  \small
  \begin{align*}
    \tpv &: \jdg\\
    \mathsf{tm}^+ &: \tpv \to \jdg\\
    \mathsf{U} &: \tpc \to \tpv\\
    \mathsf{F} &: \tpv \to \tpc\\
    \mathsf{tm}^{\ominus}(X) &\coloneqq \tmv{\UU{X}}\\
    \mathsf{ret} &: (\isof{A}{\tpv}, \isof{a}{\tmv{A}}) \to \tmc{\F{A}}\\
    \mathsf{bind} &: \impl{\isof{A}{\tpv}, \isof{X}{\tpc}} \tmc{\F{A}} \to (\tmv{A} \to \tmc{X}) \to \tmc{X}\\
    \mathsf{tbind} &: \impl{\isof{A}{\tpv}} \to \tmc{\F{A}} (\tmv{A} \to \tpc) \to \tpc\\
    \mathsf{dbind} &: \impl{\isof{A}{\tpv}, \isof{X}{\tmv{A} \tpc}}  (\isof{e}{\tmc{\F{A}}}) \to ((\isof{a}{\tmv{A}}) \to \tmc{X(a)}) \to
    \tmc{\tbind{e}{X}}
  \end{align*}
  \caption{Core \dcbpv{} calculus.}
  \label{fig:calf-cbpv}
\end{figure}

\begin{figure}[h]
  \small
  \begin{align*}
    \mathsf{bind}/\mathsf{ret} &: \impl{A,X} (\isof{a}{\tmv{A}}) \to (\isof{f}{\tmv{A} \to \tmc{X}}) \to
    \bind{\ret{a}}{f} = f(a)\\
    \mathsf{tbind}/\mathsf{ret} &: \impl{A} (\isof{a}{\tmv{A}}) \to (\isof{f}{\tmv{A} \to \tpc}) \to
    \tbind{\ret{a}}{f} = f(a)\\
    \mathsf{dbind}/\mathsf{ret} &: \impl{A,X}
      (\isof{a}{\tmv{A}}) \to (\isof{f}{\isof{a}{\tmv{A}} \to \tmc{X(a)}}) \to
      \dbind{\ret{a}}{f} = f(a)
      \\
    \mathsf{bind}/\mathsf{assoc} &: \impl{A,B,X} (\isof{e}{\tmc{\F{A}}}) \to\\
    &(\isof{f}{\tmv{A} \to \tmc{\F{B}}}) \to (\isof{g}{\tmv{B} \to \tmc{C}}) \to\\
    &\bind{\bind{e}{f}}{g} = \bind{e}{\lambda a.\, \bind{f(a)}{g}}\\
    \mathsf{tbind}/\mathsf{assoc} &: \impl{A, B, X} (\isof{e}{\tmc{\F{A}}}) \to
    (\isof{f}{\tmv{A} \to \tmc{\F{B}}}) \to\\
    &\tbind{\bind{e}{f}}{X} = \tbind{e}{\lambda a.\, \tbind{f(a)}{X}}
  \end{align*}
  \caption{Computation and associativity laws for sequencing.}
  \label{fig:calf-inversion}
\end{figure}

\begin{figure}[h]
  \small
  \begin{gather*}
    \begin{aligned}[t]
      \mathbb{C} &: \jdg\\
      0 &: \mathbb{C}\\
      + &: \mathbb{C} \to \mathbb{C} \to \mathbb{C}\\
      \le &: \mathbb{C} \to \mathbb{C} \to \jdg{}\\
      \mathsf{costMon} &: \mathsf{isCostMonoid}(\mathbb{C}, 0, +, \le)
    \end{aligned}
    \qquad
    \begin{aligned}[t]
      \widehat{\mathbb{C}} &: \tpc\\
      (\mathsf{out}_{\mathbb{C}}, \mathsf{in}_{\mathbb{C}}) &: \tmc{\widehat{\mathbb{C}}} \cong \mathbb{C}\\
      \widehat{\le} &: \widehat{\mathbb{C}} \to \widehat{\mathbb{C}} \to \tpc\\
      (\mathsf{out}_{\le}, \mathsf{in}_{\le}) &: \impl{c,c'}
        \tmc{c \mathrel{\widehat{\le}} c'} \cong (\mathsf{out}_{\mathbb{C}}(c) \le \mathsf{out}_{\mathbb{C}}(c'))
    \end{aligned}
    \\[8pt]
    \begin{aligned}[t]
      \mathsf{step} &: \impl{\isof{X}{\tpc}} \mathbb{C} \to \tmc{X} \to \tmc{X}\\
      \mathsf{step}_{0} &: \impl{X,e} \mstep{0}{e} = e\\
      \mathsf{step}_{+} &: \impl{X,e,c_1,c_2} \mstep{c_1}{\mstep{c_2}{e}} = \mstep{c_1 + c_2}{e}
    \end{aligned}
    \\[8pt]
    \begin{aligned}
      \mathsf{bind}_{\mathsf{step}} &: \impl{A,X,e,f,c} \bind{\mstep{c}{e}}{f} = \mstep{c}{\bind{e}{f}}\\
      \mathsf{tbind}_{\mathsf{step}} &: \impl{A,X,e,f,c} \tbind{\mstep{c}{e}}{f} = \tbind{e}{f}\\
      \mathsf{dbind}_{\mathsf{step}} &: \impl{A,X,e,f,c} \dbind{\mstep{c}{e}}{f} = \mstep{c}{\dbind{e}{f}}\\
    \end{aligned}
  \end{gather*}
  \caption{Cost structure and cost effect.}
  \label{fig:calf-cost}
\end{figure}

\begin{figure}[h]
  \small
  \begin{align*}
    \ExtOpn &: \jdg\\
    \ExtOpn/{\mathsf{uni}} &: \impl{\isof{u,v}{\ExtOpn}} u = v\\
    \ext{A} &\coloneqq \ExtOpn \to A\\
    \mathsf{step}/{\ExtOpn} &: \impl{X , e, c} \ext{\mstep{c}{e} = e}\\
    \Op^+ &: \tpv \to \tpv\\
    \_ &: \impl{A} \tmv{\Op^+{A}} \cong \Op(\tmv{A})\\\\
    \Cl &: \tpv \to \tpv\\
    \eta_{\Cl} &: \tmv{A} \to \tmv{\Cl{A}}\\
    \ast &: \ExtOpn \to \tmv{\Cl{A}}\\
    \_ &: \Pi \isof{a}{\tmv{A}}.\, \Pi \isof{u}{\ExtOpn}.\, \eta_{\Cl}(a) = \ast(u)\\
    \mathsf{ind}_{\Cl} &: \impl{A} (\isof{a}{\tmv{\Cl{A}}}) \to (\isof{X}{\tmv{\Cl{A}} \to \tpc}) \to\\
    &(\isof{x_0}{(\isof{a}{\tmv{A}}) \to \tmc{X(\eta_{\Cl}(a))}}) \to\\
    &(\isof{x_1}{(\isof{u}{\ExtOpn}) \to \tmc{X(\ast(u))}}) \to\\
    &((\isof{a}{\tmv{A}}) \to (\isof{u}{\ExtOpn}) \to x_0(a) = x_1(u)) \to\\
    &\tmc{X(a)}\\
    \mathsf{ind}_{\Cl}/\eta &:\impl{A} (\isof{a}{\tmv{A}}) \to (\isof{X}{\tmv{\Cl{A}} \to \tpc}) \to\\
    &(\isof{x_0}{(\isof{a}{\tmv{A}}) \to \tmc{X(\eta_{\Cl}(a))}}) \to\\
    &(\isof{x_1}{(\isof{u}{\ExtOpn}) \to \tmc{X(\ast(u))}}) \to\\
    &(\isof{h}{(\isof{a}{\tmv{A}}) \to (\isof{u}{\ExtOpn}) \to x_0(a) = x_1(u)}) \to\\
    &\mathsf{ind}_{\Cl}(\eta_{\Cl}(a), X, x_0, x_1 , h) = x_0(a)\\
    \mathsf{ind}_{\Cl}/\ast &:\impl{A} (\isof{u}{\ExtOpn}) \to (\isof{X}{\tmv{\Cl{A}} \to \tpc}) \to\\
    &(\isof{x_0}{(\isof{a}{\tmv{A}}) \to \tmc{X(\eta_{\Cl}(a))}}) \to\\
    &(\isof{x_1}{(\isof{u}{\ExtOpn}) \to \tmc{X(\ast(u))}}) \to\\
    &(\isof{h}{(\isof{a}{\tmv{A}}) \to (\isof{u}{\ExtOpn}) \to x_0(a) = x_1(u)}) \to\\
    &\mathsf{ind}_{\Cl}(\ast(u), X, x_0, x_1 , h) = x_1(u)
  \end{align*}
  \caption{Modal account of the phase distinction.}
  \label{fig:calf-phaseDistinction}
\end{figure}

\begin{figure}[h]
  \small
  \begin{align*}
    \Pi &: (\isof{A}{\tpv}, \isof{X}{\tmv{A} \to \tpc}) \to \tpc\\
    (\mathsf{ap}, \mathsf{lam}) &: \impl{A,X} \tmc{\Pi(A; X)} \cong (\isof{a}{\tmv{A}}) \to \tmc{X(a)}\\\\
    \Sigma^{++} &: (\isof{A}{\tpv}, \isof{B}{\tmv{A} \to \tpv}) \to \tpv\\
    (\mathsf{unpair}^{++}, \mathsf{pair}^{++}) &: \impl{A,B} \tmv{\sigpos{A}{B}} \cong \Sigma (\tmv{A}) (\lambda a.\, \tmv{B(a)})\\
    \Sigma^{+-} &: (\isof{A}{\tpv}, \isof{X}{\tmv{A} \to \tpc}) \to \tpc\\
    (\mathsf{unpair}^{+-}, \mathsf{pair}^{+-}) &: \impl{A,X} \tmc{\signeg{A}{X}} \cong \Sigma (\tmv{A}) (\lambda a.\, \tmc{X(a)})\\\\
    \mathsf{eq} &: (\isof{A}{\tpv}) \to \tmv{A} \to \tmv{A} \to \tpv\\
    \mathsf{self} &: \impl{A} (\isof{a,b}{\tmv{A}}) \to a =_{\tmv{A}} b \to \tmv{\eqty{A}{a}{b}}\\
    \mathsf{ref} &: \impl{A} (\isof{a,b}{\tmv{A}}) \to \tmc{\F{\eqty{A}{a}{b}}} \to a =_{\tmv{A}} b\\
    \mathsf{uni} &: \impl{A,a,b} (\isof{p,q}{\tmc{\F{\eqty{A}{a}{b}}}}) \to \Op{(p = q)}\\\\
    \unit &: \tpv\\
    \triv &: \tmv{\unit}\\
    \eta_{\unit} &: \impl{u,v} u =_{\tmv{\unit}} v\\\\
    + &: \tpv \to \tpv \to \tpv\\
    \mathsf{inl} &: \impl{A,B} \tmv{A} \to \tmv{\sumty{A}{B}}\\
    \mathsf{inr} &: \impl{A,B} \tmv{B} \to \tmv{\sumty{A}{B}}\\
    \mathsf{case} &: \impl{A,B} (\isof{s}{\tmv{\sumty{A}{B}}}) \to (\tmv{\sumty{A}{B}} \to \tpc) \to\\
    &((\isof{a}{\tmv{A}}) \to \tmc{X(\inl{a})}) \to\\
    &((\isof{b}{\tmv{B}}) \to \tmc{X(\inl{b})}) \to
    \tmc{X(s)} \\
    \mathsf{case}_{\mathsf{inl}} &: \impl{A,B,X,e_0,e_1} (\isof{a}{\tmv{A}}) \to
    \sumcase{\inl{a}}{X}{e_0}{e_1} = e_0(a)\\
    \mathsf{case}_{\mathsf{inr}} &: \impl{A,B,X,e_0,e_1} (\isof{b}{\tmv{B}}) \to
    \sumcase{\inr{b}}{X}{e_0}{e_1} = e_1(b)\\\\
    \nat &: \tpv\\
    \zero &: \tmv{\nat}\\
    \mathsf{suc} &: \tmv{\nat} \to \tmv{\nat}\\
    \mathsf{rec} &: (\isof{n}{\tmv{\nat}}) \to (\isof{X}{\tmv{\nat} \to \tpc}) \to \tmc{X(\zero)} \to\\
      &((\isof{n}{\tmv{\nat}}) \to \tmc{X(n)} \to \tmc{X(\suc{n})}) \to \tmc{X(n)}\\
    \mathsf{rec}/\mathsf{zero} &: \impl{X,e_0,e_1} \rec{\zero}{X}{e_0}{e_1} = e_0\\
    \mathsf{rec}/\mathsf{suc} &: \impl{n,X,e_0,e_1} \rec{\suc{n}}{X}{e_0}{e_1} = e_1(n)(\rec{n}{X}{e_0}{e_1})
  \end{align*}
  \caption{Types}
  \label{fig:calf-types}
\end{figure}

\begin{figure}[h]
  \small
  \begin{align*}
    \mathsf{L} &: \Alert{\mathbb{C}} \to \tpv \to \tpv\\
    \nilex &: \impl{c,A} \calist{c}{A}\\
    \mathsf{cons} &: \impl{c,A} \tmv{A} \to \tmv{\calist{c}{A}} \to \tmv{\calist{c}{A}}\\
    \mathsf{rec} &: \impl{c,A} (\isof{l}{\tmv{\calist{c}{A}}}) \to (\isof{X}{\tmv{\calist{c}{A}} \to \tpc}) \to\\
    &(\tmc{X(\nilex)}) \to \\
    &((\isof{a}{\tmv{A}}) \to (\isof{l}{\tmv{\calist{c}{A}}}) \to \tmc{X(l)} \to \tmc{X(\consex{a}{l})}) \to \\
    &\tmc{X(l)}\\
    \mathsf{rec}/\mathsf{nil} &: \impl{c,A,X,e_0,e_1} \listrec{\nilex}{X}{e_0}{e_1} = e_0\\
    \mathsf{rec}/\mathsf{cons} &: \impl{c,A,a,X,e_0,e_1} (\isof{l}{\tmv{\calist{c}{A}}}) \to\\
    &\listrec{\consex{a}{l}}{X}{e_0}{e_1} = \Alert{\mathsf{step}^c}(e_1(a)(l)(\listrec{l}{X}{e_0}{e_1}))
  \end{align*}
  \caption{Types, continued}
  \label{fig:calf-types-1}
\end{figure}

\begin{figure}[h]
  \small
  \begin{align*}
    \mathsf{lam}_{\mathsf{step}} &: \impl{A,X,c} (\isof{f}{(\isof{a}{\tmv{A}}) \to \tmc{X(a)}}) \to \mstep{c}{\lam{f}} = \lam{\mstep{c}{f}}\\
    \mathsf{pair}^{+-}_{\mathsf{step}} &: \impl{A,X,c} (\isof{e}{\Sigma (\tmv{A})(\lambda a.\, \tmc{X(a)})}) \to \mstep{c}{e} = (\fst{e} , \mstep{c}{\snd{e}})\\
    \mathsf{case}_{\mathsf{step}} &: \impl{A,B,X,e_0,e_1,c} (\isof{s}{\tmv{A + B}}) \to\\
    &\mstep{c}{\sumcase{s}{X}{e_0}{e_1}} = \sumcase{s}{X}{\lambda a.\, \mstep{c}{e_0(a)}}{\lambda b.\, \mstep{c}{e_1(b)}}
  \end{align*}
  \caption{Interaction of $\mathsf{step}$ with type structure.}
  \label{fig:calf-types-step}
\end{figure}

\begin{figure}[h]
  \small
  \begin{align*}
    \& &: \impl{\isof{A,B}{\tpv}} \tmc{\F{A}} \to \tmc{\F{B}} \to \tmc{\F{A \times B}}\\
    \&_{\mathsf{join}} &: \impl{A,B,c_1,c_2,a,b}
    \para{\left(\mstep{c_1}{\ret{a}}\right)}{\left(\mstep{c_2}{\ret{b}}\right)} = \mstep{c_1 \otimes c_2}{\ret{(a,b)}}
  \end{align*}
  \caption{Parallelism.}
  \label{fig:calf-para}
\end{figure}

\end{document}